\documentclass[10pt,twocolumn,UKenglish]{article}

\usepackage{blindtext, graphicx}
\usepackage{balance}

\usepackage{amsthm,amsmath,ed,amssymb,mathtools,fixmath}
\usepackage{algorithmicx}
\usepackage{fullpage}
\usepackage{subcaption}
\usepackage{algorithm}
\usepackage{algpseudocode}
\usepackage{xcolor}
\usepackage{array}
\usepackage[]{hyperref}
\newtheorem{theorem}{Theorem}

\newtheorem{lemma}{Lemma}
\newtheorem{property}{Property}
\newtheorem{proposition}{Proposition}

\newtheorem{dfn}{Definition}
\newtheorem{clm}{Claim}

\renewcommand{\paragraph}[1]{\vskip3mm\noindent\textbf{#1}}

\newcommand{\old}[1]{\iffalse #1 \fi}
\newcommand{\mat}[1]{\mathbold{#1}}
\renewcommand{\vec}[1]{\mathbold{#1}}
\newcommand{\1}{\mathbf{1}}
\DeclareMathOperator*{\argmin}{arg\,min}

\newcommand{\System}{\emph{CACD}\xspace}

\newcommand{\E}[1]{\mathbb{E}\hspace{-0.5mm}\left( #1 \right)}
\renewcommand{\P}[1]{\mathbf{P}\hspace{-0.5mm}\left( #1 \right)}

\newcommand{\svr}{node\xspace}
\newcommand{\svrs}{nodes\xspace}
\newcommand{\serv}{u\xspace}
\newcommand{\svi}[1]{\serv_{#1}}

\newcommand{\M}{\mathcal{R}}

\newcommand{\ledge}[1]{l(#1)}
\newcommand{\redge}[1]{r(#1)}
\newcommand{\bedge}[1]{b(#1)}

\newcommand{\wlg}{w.\,l.\,o.\,g.\/~}

\newcommand{\sv}[1]{s_{#1}}
\newcommand{\cw}[1]{\textit{cw}_{#1}}
\newcommand{\cs}[1]{\textit{cs}_{#1}}

\newcommand{\sigi}[1]{\sigma(#1)}

\newcommand{\w}[1]{\mathrm{walk}(#1)}

\newcommand{\Rmnum}[1]{\expandafter\@slowromancap\romannumeral #1@}

\def\abs#1{\lvert #1 \rvert}

\newcommand{\concat}[2]{#1 \oplus #2}
\begin{document}

\title{Towards Entropy-Proportional Robust Topologies} 

\title{Towards Communication-Aware Robust Topologies\thanks{Research
supported by the German-Israeli Foundation for Scientific Research and Development (GIF), grant number I-1245-407.6/2014.}}





\author{Chen Avin$^1$ \quad Alexandr Hercules$^1$ 
\quad Andreas Loukas$^2$ \quad Stefan Schmid$^3$\\
{\small $^1$ Ben Gurion University, Israel \quad  $^2$ EPFL, Switzerland \quad 
 $^3$ Aalborg University, Denmark}}

\date{}

\maketitle

\begin{abstract}
We currently witness the emergence of 
interesting new network topologies optimized towards
the traffic matrices they serve, 
such as demand-aware datacenter interconnects 
(e.g., ProjecToR) and demand-aware
overlay networks (e.g., SplayNets).
This paper introduces a formal framework and approach 
to reason about and design such topologies. We leverage a connection between
the communication frequency of two nodes and the path length between them in the network, which depends on the \emph{entropy} of the communication matrix. 
Our main contribution is a novel robust, yet sparse,
family of network topologies which guarantee an expected path length 
that is proportional to the entropy of the communication patterns.
\end{abstract}



\vspace{-3mm}
\section{Introduction}\label{sec:intro}

Traditionally, the topologies and interconnects of computer networks
were optimized toward static worst-case criteria, such as maximal
diameter, maximal degree, or bisection bandwidth. For example, 
many modern datacenter interconnects are based on Clos topologies~\cite{comm-dc}
which provide a constant diameter and a high bisection bandwidth,
and overlay networks are often hypercubic, providing a logarithmic
degree and route-length in the worst case.

While such topologies are efficient in case of worst-case traffic
patterns, researchers have recently started developing novel
network topologies which are optimized towards
the traffic matrices which they actually serve, henceforth
referred to as \emph{demand-aware topologies}. 
For example, ProjecToR (ACM SIGCOMM 2016~\cite{projector}) 
describes a novel datacenter interconnect based on
laser-photodetector edges, which can be established according
to the served traffic pattern, which often are far from random
but exhibit locality. Another example are 
SplayNet overlays (IEEE/ACM Trans.~Netw.~2016~\cite{splaynet}).

\subsection{Our Contribution}

We in this paper present a novel approach
to design demand-aware topologies, which
come with provable performance and robustness
guarantees. We consider 
a natural new metric to measure the quality
of a demand-optimized network topology,
namely whether the provided path lengths 
are proportional to the \emph{entropy} in the traffic
matrix: frequently communicating nodes should
be located closer to each other. 
Entropy is a well-known
metric in information and coding theory,
and indeed, the topology designs
presented in this paper are based on coding theory.

Our main contribution is a novel robust and sparse
family of network topologies which guarantee an entropy-proportional
path length distribution.
Our approach builds upon the  
continuous-discrete design introduced 
by Naor and Wieder~\cite{naor2007novel}. 
The two key benefits of the continuous-discrete design
are its flexibility and simplicity:
It allows 
to formally reason about
topologies as well as routing schemes in the 
continuous space, and a simple discretization 
results in network topologies which preserve 
the derived guarantees. 

At the heart of our approach lies a novel coding-based approach, 
which may be of independent interest. 
 
We also believe that our work offers interesting new insights
 into the classic continuous-discrete approach. 
 For instance, we observe that greedy routing can combine both 
 \emph{forward} and \emph{backward} routing, which introduces
 additional flexibilities.

\subsection{Related Work}

We are not aware of any work on robust
network topologies providing 
entropy-proportional path length guarantees.
Moreover, to the best of our knowledge, there is no 
work on continuous-discrete network
designs for non-uniform distribution probabilities.

In general, however, the design of efficient communication networks has
been studied extensively for many years already.
While in the 1990s, efficient network topologies were studied
intensively in the context of VLSI designs~\cite{Leighton:1991:IPA:119339},
in the 2000s, researchers were especially interested in designing
peer-to-peer overlays~\cite{stoica2001chord,aspnes2007skip,malkhi2002viceroy}, and more recently,
research has, in particular, focused on designing efficient
datacenter interconnects~\cite{dcell,bcube}.

The design of robust networks especially has been of prime interest
to researchers for many decades. In the field of network design,
survivable networks are explored that remain functional when links
are severed or nodes fail, that is, network services can be
restored in the event of catastrophic failures~\cite{rob-net-des-1,rob-net-des-2}.
Recently, researchers have proposed an interesting generic and declarative approach
to network design~\cite{condor}.

The motivation for our work comes from recent
advances in more flexible network designs,
also leveraging
the often non-uniform traffic demands~\cite{traffic-1},
most notably ProjecToR~\cite{projector}, but also projects related in spirit such as Helios~\cite{ref16-helios}, REACToR~\cite{ref25-reactor}, Flyways~\cite{ref23-flyways}, Mirror~\cite{zhou2012mirror}, Firefly~\cite{ref22-firefly}, etc.~have explored the use of several
underlying technologies to build such networks. 

The works closest to ours are the Continuous-Discrete 
approach by Naor and Wieder~\cite{naor2007novel}, the  
SplayNet~\cite{splaynet} paper and the work by
Avin et al.~\cite{disc17netdesign}. 
We tailor the first~\cite{naor2007novel} to demand-optimized networks,
and provide new insights e.g., 
 how greedy routing can be used to combine both 
 \emph{forward} and \emph{backward} routing,  
 introducing
 additional flexibilities. 
SplayNet~\cite{splaynet} and Avin et al.~\cite{disc17netdesign}
focus on binary search trees
resp.~on sparse constant-degree network designs,
however, they do not provide any robustness or path diversity
guarantees. Moreover, ~\cite{disc17netdesign} does
not provide local routing.

\subsection{Organization}

The remainder of this paper is organized as follows.
In Section~\ref{sec:model} we describe our model.
Section~\ref{sec:prelim} introduces  
background and preliminaries.
Section~\ref{sec:system} then presents and formally
evaluates our
coding approach to topology designs.
We then present routing and network properties in Sections \ref{sec:algorithms} and \ref{sec:properties}.
After reporting on simulation results
in Section~\ref{sec:sims}
we conclude our contribution in Section~\ref{sec:conclusion}.

\section{The Problem}\label{sec:model}




Our objective is to design a communication network $G = (V,E)$, where $V$ is the set of nodes and $E$ is the set of edges, with $n = \abs{V}$ and $m = \abs{E}$. 

%
We assume that the network needs to serve route requests that are 
drawn independently from an arbitrary, but fixed distribution. 
We represent this request distribution $\mathcal{R}$ using the $n\times n$ \emph{demand} 
matrix $\mat{R}$, where entry $R_{ij}$ denotes 
the probability of a request from source $u_i$ to destination $u_j$.
The \emph{source} and \emph{destination} marginal distribution vectors are defined, respectively, as 
\begin{align}
	\vec{p_s} = \mat{R} \1 \quad \text{and} \quad \vec{p_d}^\top = \1^\top \mat{R}, 
\end{align}
where $\1$ is the all-ones vector and, being a probability matrix, $\1^\top \mat{R} \1 = 1$.
%
%


While it is desirable that $G$ has a low diameter, 
we in this paper
additionally require that this diameter is also found, by a simple 
 \emph{greedy routing} scheme. Supporting greedy routing
 is attractive, as 
messages can always be forwarded without global knowledge, but only 
based on the header and neighbor information. 
This allows for a more scalable forwarding and also supports
fast implementations, also in hardware. 

Given a network $G$ and a routing algorithm $\mathcal{A}$, 
denote by $Route_{G, \mathcal{A}}(i,j)$ the route length 
from node $u_i$ to node $u_j$. 
Traditionally, the route lengths are optimized 
uniformly across all possible pairs. However,
in this paper, we seek to optimize the 
\emph{expected path length} 
\begin{dfn}
    Given a request distribution $\M$, a network $G$ and a routing algorithm $\mathcal{A}$, 
    the expected path length of the system is defined as:
\begin{align}\label{eq:EPL}
       \text{EPL}(\mathcal{R}, G, \mathcal{A}) &= \mathbb{E}_{\mathcal{R}}[\textit{Route}_{G, \mathcal{A}}] \nonumber \\
    &= \sum\limits_{u_i,u_j \in V} R_{ij} \cdot Route_{G, \mathcal{A}}(i,j)
\end{align}
(or $\text{EPL}$ for short), with respect to distribution $\mathcal{R}$. 
\end{dfn}

%


We seek to design network topologies which
are not only good for routing, but also feature
desirable properties along other dimensions. 
These dimensions are often not binary but offer   
a spectrum, and sometimes also stand in conflict with
other properties, introducing tradeoffs. 
We consider the following four basic dimensions:
\begin{enumerate}

\item \textbf{Routing Efficiency:} The network should feature a small expected path length. Therefore, the diameter, and more importantly, the routing distance (which may be larger) between frequently communicating nodes should be shorter than the distance between nodes that communicate less frequently. 

\item \textbf{Sparsity:} 
We are interested in sparse networks, where the number of edges is 
\emph{small}, ideally even linearly proportional to the number of nodes \svrs. 
Indeed, as we will see, our proposed graphs have
at most $|E| \leq 4n$ edges. 

\item \textbf{Fairness:} The degree of a node in the network as well as the node's 
level of participation in the routing, should be proportional to the node's 
activity level in the network.
Putting it differently, a network is \emph{unfair} if nodes that are not very 
active (in $\mathcal{R}$) have a high degree or need to forward a disproportional
amount of traffic. 
To be more formal, let $d_i$ denote the degree of node $u_i$ in $G_\mathcal{R}$ and in addition, define the $n\times 1$ \emph{activity vector} $\vec{a} = \vec{p_s} + \vec{p_d}$, where element $a_i$ is the activity level of $u_i$ drawn from the request distribution $\mathcal{R}$. We require that the degree of $u_i$ is linearly proportional to $a_i$,
that is, a small positive constant $c$ exists such that $d_i \leq c \, a_i \, n$ for all nodes~$u_i$.

\item \textbf{Robustness:} The network and  the routing algorithm should be robust to (random and adversarial)
edge failures. 
\end{enumerate}


Let's consider an example demand matrix along with different possible network topologies
(see Figure \ref{fig:net_examplest}, 
the size of nodes is proportional to their activity level).
Assume a network with 9 \svrs $V=\svi{0}, \dotsc, \svi{8}$ and 
an activity level distribution $\vec{p} = \{p_1, p_2, \dotsc, p_9\}$
where $p_i = c^{-1} \cdot i^{-1}$ and $\sum p_i=1$. Additionally assume that nodes 
$i$ and $j$ communicate with probability $p(i,j) = p_i \cdot p_j$. 
Let's consider a \emph{complete graph}. A complete graph topology is very \emph{robust} to failures and \emph{efficient} in terms of routing. 
But the degree of every node is high, independently of
the activity, and the number of edges is $n^2$. Thus,
the network is not \emph{sparse} and not \emph{fair}. 


Our second network is 
a \emph{star} network. The star topology is 
 \emph{sparse} with only $n-1$ edges. The \emph{efficiency} is almost the same 
 as for complete network, that largest route has length 2. 
Also the edge robustness is good since edge failures will lead to a disconnection 
of the same number of \svrs; but a single failure of the central \svr 
will bring down the entire network, so this topology is not \emph{robust}.
The \emph{fairness} of the network is also problematic: 
the central node has degree $n-1$
and is involved in all routing, although its activity level may be much smaller than that.

\begin{figure}
	\begin{tabular}{cc}
		\includegraphics[width=.45\columnwidth]{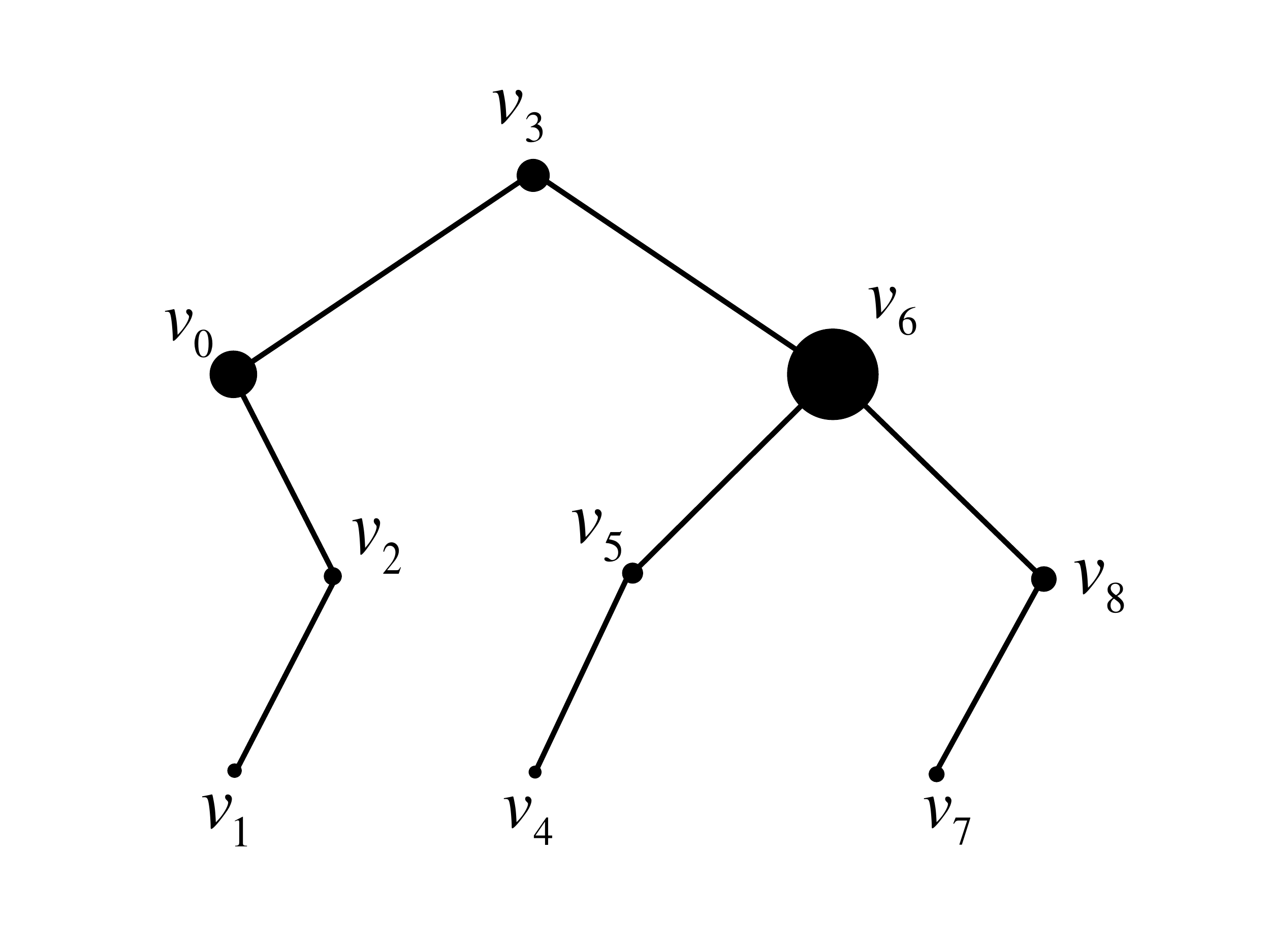} &
		\includegraphics[width=.45\columnwidth]{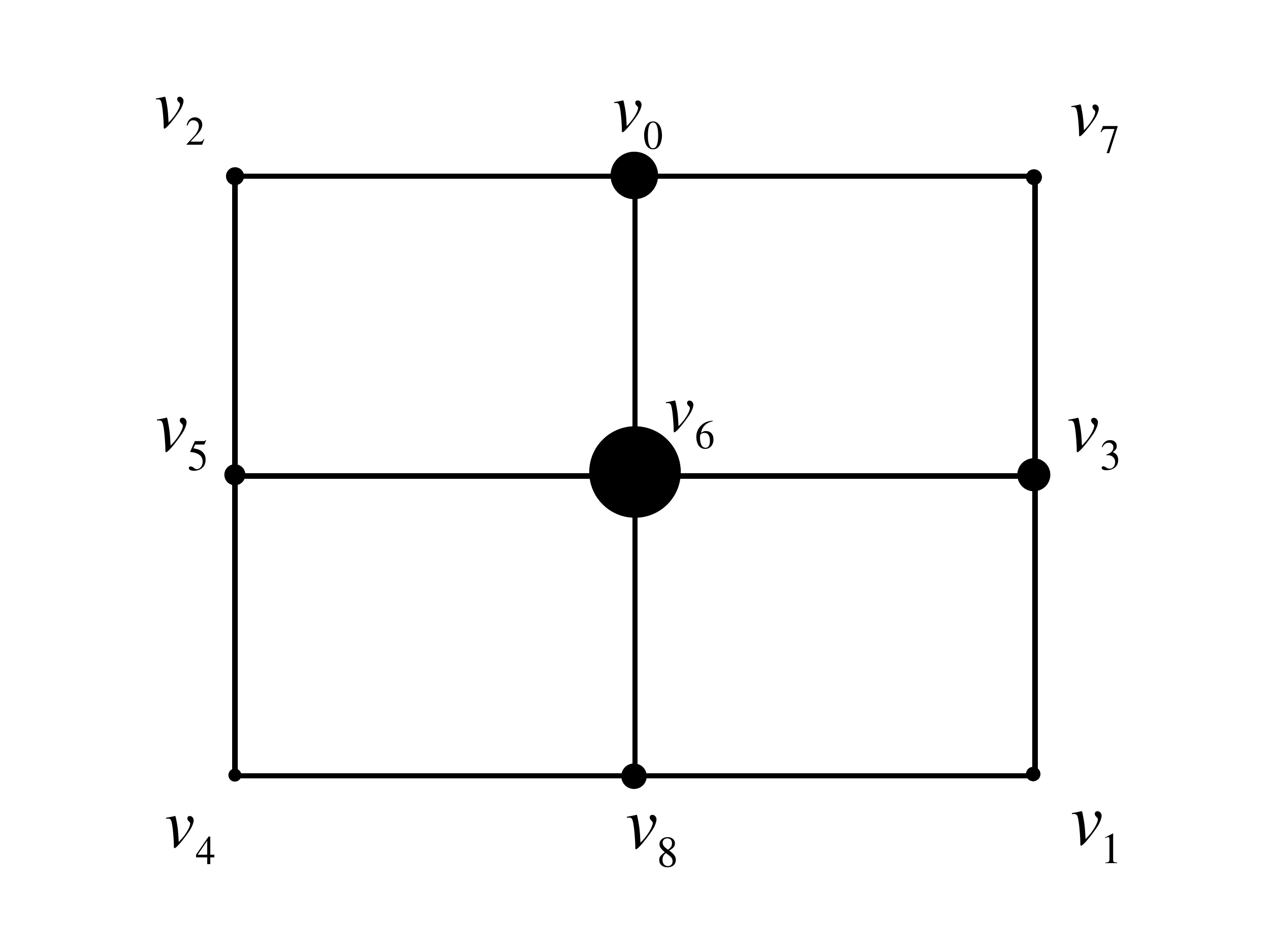} \\
		(a) binary search tree net & (b) a grid network \\
		\includegraphics[width=.45\columnwidth]{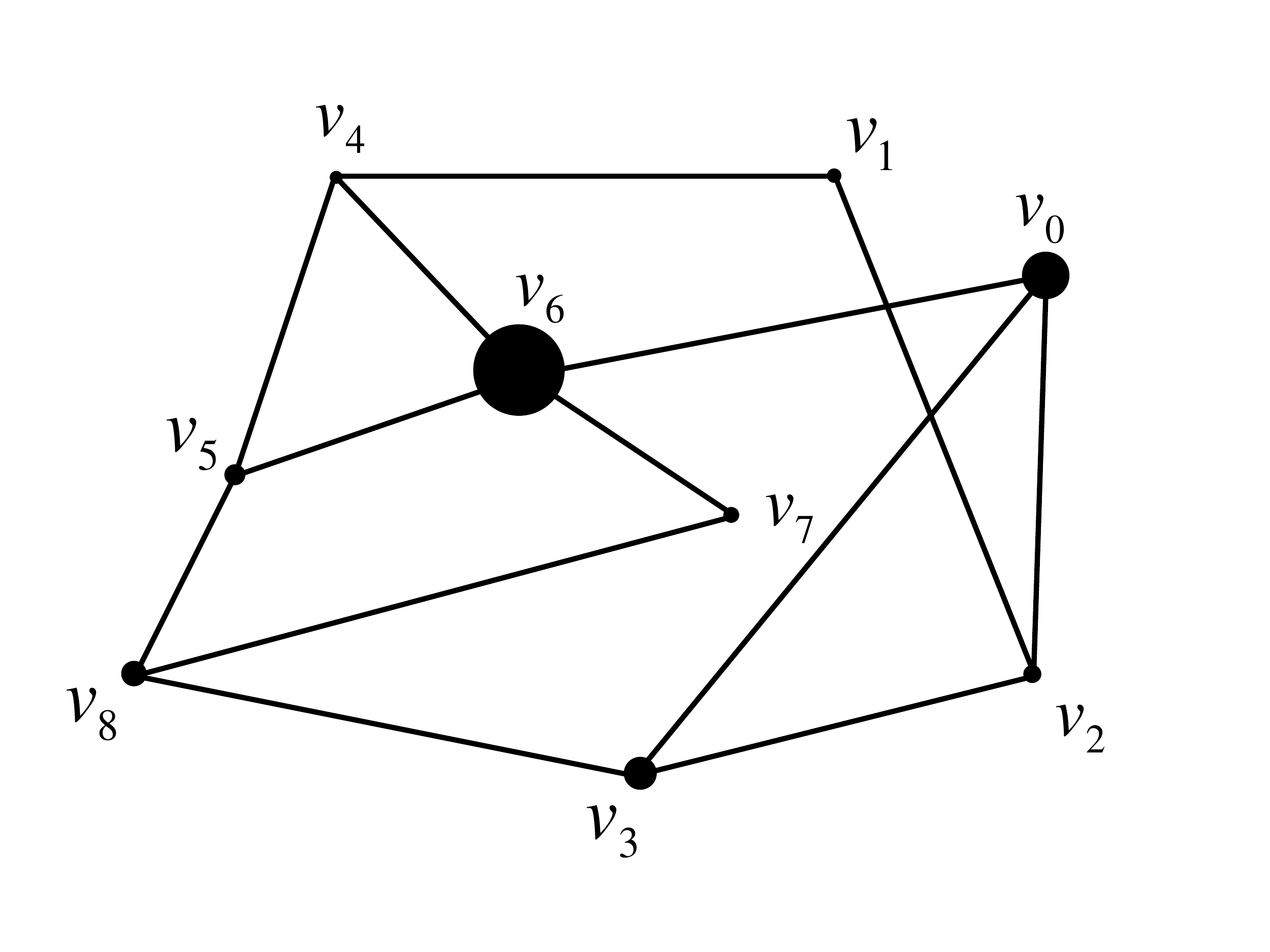} &
		\includegraphics[width=.45\columnwidth]{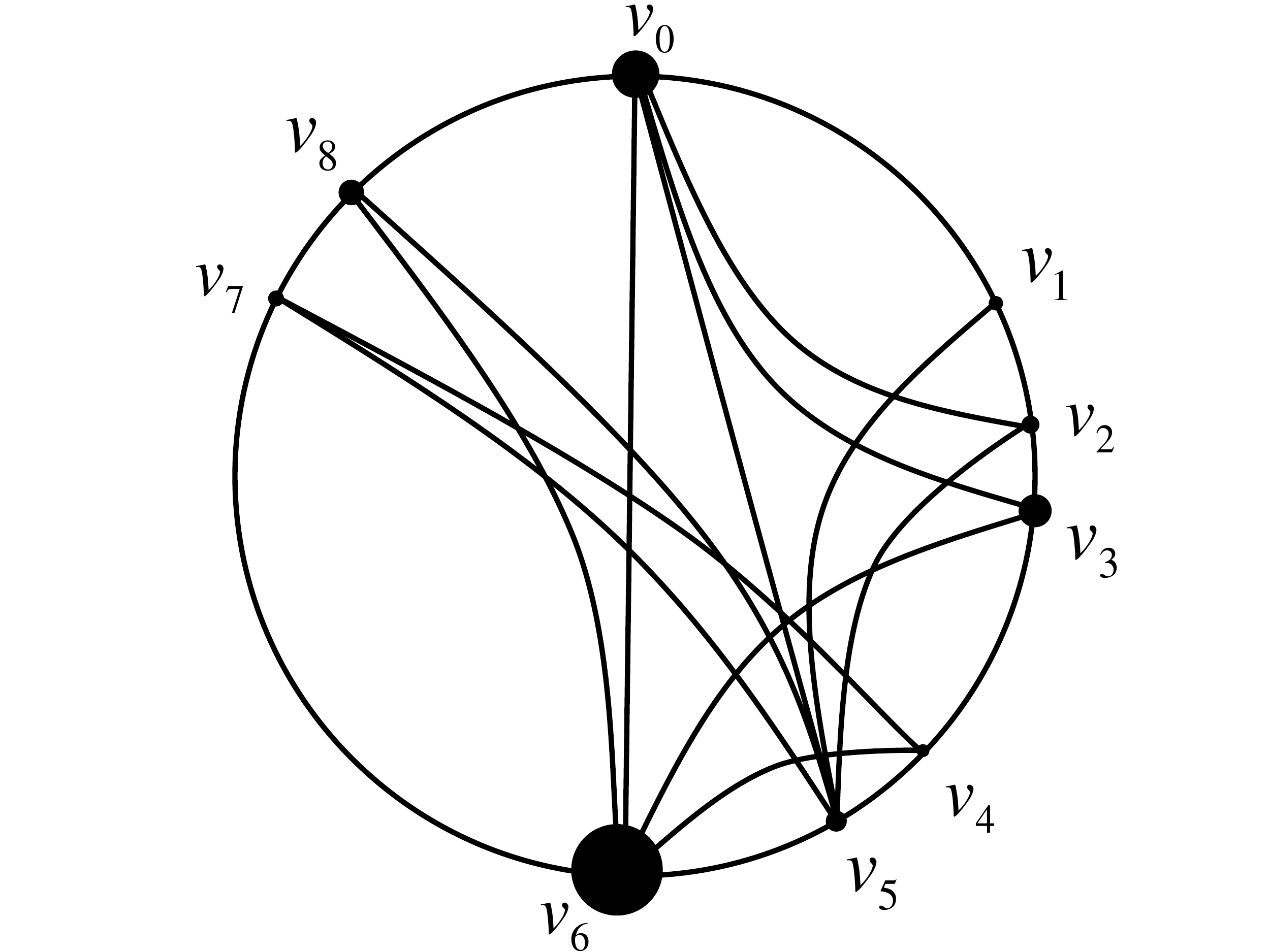} \\
		 (c) random network & (d) \System network \\
		\end{tabular}
\caption{Example network designs.}
\label{fig:net_examplest}
\end{figure}

To solve the fairness problem, but keep the graph sparse, researchers recently 
proposed a communication aware (overlay) network 
that is based on a binary search tree~\cite{splaynet}. Since the network is a 
binary search tree, all nodes have low degree, 
the graph is \emph{sparse} and achieves \emph{fairness}
(with regard to node degree) as well as greedy routing. 
Beside having self-adujsting properties (which we do not consider here), 
the authors showed that the expected route length in the network is proportional 
to the \emph{maximum} of the entropy of the source and destination distributions. 
This is achieved by placing highly active nodes closer to the root. Moreover the authors showed that their 
bound on the route length is optimal in the case that the request 
distribution follows a product distribution (and the network is a binary tree). 
However, the network is a tree, and is not robust.


A topology that overcomes the robustness problem of trees, but keeps the graph sparse, 
is the \emph{grid} (see Figure \ref{fig:net_examplest} (d)).
This topology is \emph{fair}, because each node keeps only a small number, up to 4, 
transport connections to other nodes. Moreover, this topology is \emph{robust}:
it allows for multiple routing paths between any two nodes, therefore in case of 
an edge or node failure, the graph stays connected. The number of edges 
is linearly proportional to the number of \svrs, so it's \emph{sparse}. Additionally the grid supports greedy routing, but its major problem is that the routing paths (and the shortest paths) are long, thus its not \emph{efficient} in terms of expected path length.

We continue our discussion by considering a network that is \emph{sparse}, \emph{robust}, \emph{fair} and has short paths between nodes.
An example are random graphs. It is well known that random sparse graphs will have all the above properties,
for example by using the famous Erd\H{o}s-R\'enyi random graph model~\cite{erdos1959on-random} (see Figure \ref{fig:net_examplest} (e)).
A known problem with random graphs is that they do not support greedy routing \cite{kleinberg2000small}, and these graphs are hence 
not interesting in our case.


We end the discussion by reaching our proposed topology designs, \System,
which is shown in Figure \ref{fig:net_examplest} (f). We will show, that our solution provides a
good balance between all the above properties. 
It is a \emph{sparse}, \emph{fair} and \emph{robust} network. And as required,
it support greedy routing while being communication aware by design. 
That is, the expected route length
is a function of the communication matrix $\M$. It particular, our main result proves that the expected path length is the \emph{minimum}  
of the entropy of the source and destination distributions.

\section{Preliminaries}\label{sec:prelim}

In order to design communication-aware network topologies,
we pursue an information-theoretical approach.
This section revisits some basic information-theoretic concepts which are important
to understand the remainder of the paper. Moreover, it provides
the necessary background on topology designs, and introduces
preliminaries.

%
%

\paragraph{Entropy and Shannon-Fano-Elias Coding.}
Recall that entropy is a measure of unpredictability of information content~\cite{shannon2001mathematical}. For a discrete random variable $X$ with possible values
$\{x_1, \dots , x_n\}$, the (binary) entropy $H(X)$ of $X$ is defined as
$H(X) = \sum_{i=1}^n p_i\log_2\frac{1}{p_i}$
where $p_i$ is the probability that $X$ takes the value $x_i$. Note that, $0 \cdot \log_2\frac{1}{0}=0$ and we usually assume that
$p_i > 0$ $\forall i$. Let $\vec{p}$ denote $X$'s probability distribution,
 then we may write $H(\vec{p})$ instead of $H(X)$.

Shannon-Fano-Elias~\cite{thomas2006elements} is a well-known \emph{prefix code} for lossless data compression. 
The Shannon-Fano-Elias algorithm derives variable-length codewords
based on the  probability of each possible symbol, depending 
on its estimated frequency of occurrence. 
As in other entropy encoding methods, more common symbols are generally represented using fewer bits, to reduce the expected code length.
We choose this coding method for our topology design, due to its simplicity and since its expected code length 
is almost optimal.


Consider a discrete random variable of \emph{symbols} $X$ with possible values $\{x_1, \dots , x_n\}$ and the corresponding symbol probability $p_i$. 
The encoding is based on cumulative distribution function (CDF)
$F_i = F(x_i) = \sum\limits_{j \leq i} p_j$.
%
The coding scheme encodes symbols using function 
\begin{align}\label{eq:cwfunc}
	\bar{F}_i = \sum_{j < i} p_j + \frac{p_i}{2} = F_{i-1} + \frac{p_i}{2}, 
\end{align}
%
%
Denote by $(x)_{01}$ the binary representation of $x$. The codeword for symbol $x_i$ consists of the first $\ell_i$ bits of the fractional part of $\bar{F}_i$, 
$	cw_i = \lfloor (\bar{F}_i)_{01} \rfloor_{\ell_i},
$%
where the \emph{code length} $\ell_i$ is defined as 
\begin{align}\label{eq:cwlength}
\ell_i = \lceil \log ({p_i}^{-1}) \rceil + 1. 
\end{align}

%
%

The above construction guarantees $(i)$ that the codewords $cw_i$ are are prefix-free 
and $(ii)$ that the expected code length 
\begin{align}\label{eq:ecl}
L_{SFE}(X) = \sum\limits_{i = 1}^n p_i \cdot l_i &= \sum\limits_{i = 1}^n p_i (\lceil \log{({p_i}^{-1})} \rceil + 1) 
\end{align}
is close to the entropy $H(X)$ of random variable $X$
\begin{align}\label{eq:eclbound}
	H(X) + 1 \leq L_{SFE}(X) < H(X) + 2,
\end{align}

\paragraph{Continuous-Discrete Approach.}
Our work builds upon the continuous-discrete topology design approach
introduced by Naor and Wieder~\cite{naor2007novel}. 
It is based on a discretization of 
a continuous space into segments, corresponding to nodes. 
There are two variants, and here we follow  
the first variant which is a Distributed Hash Table (DHT),
called \emph{Distance Halving}.

The construction starts with a \emph{continuous graph} $G_c$ 
defined over a 1-dimensional cyclic space $I = [0, 1)$. 
%
As shown in Figure \ref{fig:Edges}, for every point $x \in I$, the \emph{left}, \emph{right}, and \emph{backward} edges of $x$ are the points 
\begin{align*}
	\ledge{x} = \frac{x}{2}, \quad \redge{x} = \frac{x+1}{2}, \quad \bedge{x} = 2y \, \text{mod}\, 1.
\end{align*} 
It is easy to check that $\ledge{x}$ and $\redge{x}$ always fall in $[0,1)$. In addition, when $x$ is written in binary form, $\ledge{x}$ effectively inserts a 0 at the left (most significant bit), whereas $\redge{x}$ shifts a 1 into the left. The backward edge removes the most significant bit.

The \emph{discrete network} $G_\vec{x}$ is then a discretization of $G_c$ 
according to a set of $n$ points $\vec{x}$ in $I$, with $x_i < x_{i+1}$ for all $i$.
The points of $\vec{x}$ divide $I$ into $n$ segments, one for each node: 
\begin{align*}
	s_i = [x_i, x_{i+1}) \, \forall i<n \quad \text{and} \quad s_n = [x_{n-1}, 1) \cup  [0, x_1)
\end{align*}
%
%
Nodes $x_i, x_j$ are connected by an edge in the discrete graph $G$ if there exists an edge $(y, z)$ in the continuous graph, such that $y \in s_i$ and $z \in s_j$. 
In addition, we add edges $(x_i, x_{i+1})$ and $(x_{n-1}, x_0)$ so that $G_{\vec{x}}$ 
contains a ring.


%

\begin{figure}[t!]
    \centering
    \includegraphics[width=250pt]{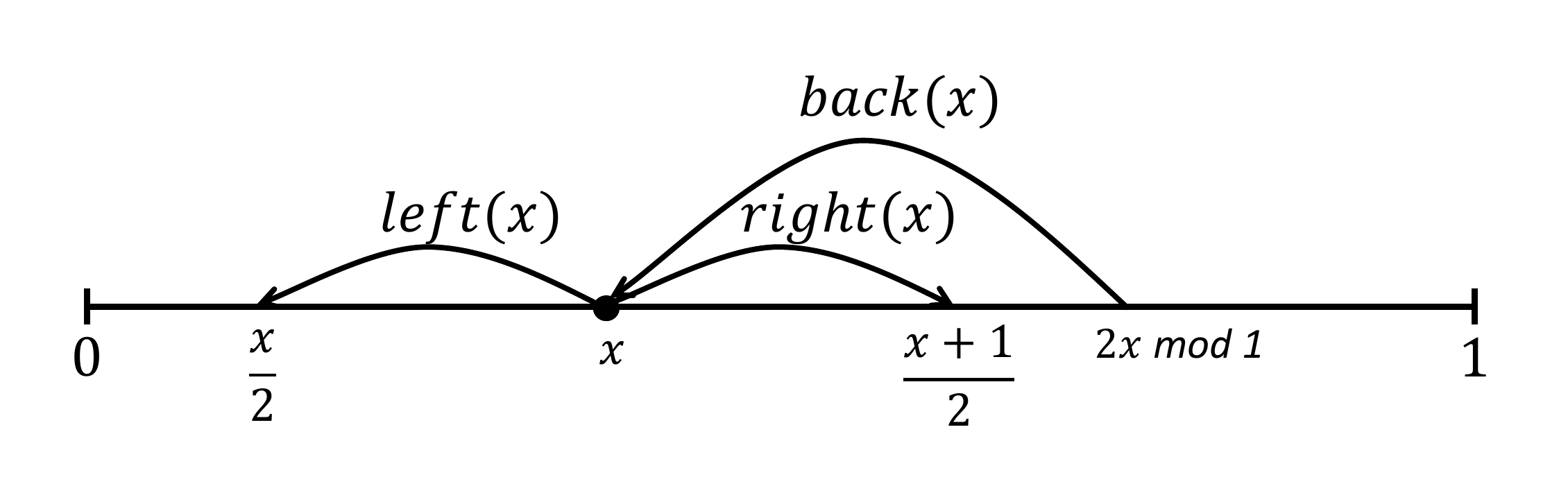}
    \caption{The edges of a point in the continuous graph $G_c$.}
    \label{fig:Edges}
\end{figure}


An important parameter of the decomposition of $I$ is the ratio between the size of the largest and the smallest cell (segment), which is called the \emph{smoothness} and denoted by $\rho(\vec{x}) = \max_{i,j}\frac{\left | s_i \right |}{\left | s_j \right |}$. The total number of edges in $G_{\vec{x}}$ without the ring edges is at most $3n-1$, the maximum out-degree without the ring edges is at most $\rho(\vec{x}) + 4$, and the maximum in-degree without the ring edges is at most $\lceil2\rho(\vec{x})\rceil + 1$.

In the original construction by Naor and Wieder, the $x_i$s were assumed to be uniform random variables. The goal 
is to offer a constant degree network with equal loads, and ensure smoothness (i.e., minimal $\rho$). 
The authors also show that the Distance Halving construction resembles the well known De Bruijn graphs~\cite{bruijn1946combinatorial}:
if $x_i = \frac{i}{n}$ and $n=2^r$ then the discrete Distance Halving graph $G_{\vec{x}}$ without the ring edges is isomorphic to the $r$-dimensional De Bruijn graph.
Based on this the authors propose two greedy lookup algorithms with a path length of logarithmic order (i.e., $r$). We use similar ideas in our routing.

\section{\System Topology Design}\label{sec:system}

We propose a coding-based topology design
which reflects communication patterns. 
We will show that our solution provides 
an efficient routing (the expected path length is the \emph{minimum}  
of the source and destination distribution entropy),
but also meets
our requirements in terms of sparsety,
fairness and robustness. 


The basic idea behind our communication-aware
continuous-discrete (\System) topology design is simple. 
Similar to the classic continuous-discrete approach, 
we start by designing a continuous network $G_c$ in the 1-dimensional cyclic 
space $I = [0, 1)$. This continuous network is subsequently discretized 
so as to obtain $G$. In contrast to previous work however, in our {discrete} 
graph construction we do not place nodes (or, more precisely, points $\vec{x}$) in a 
uniform or random manner. 
Rather, the points $\vec{x}$ are placed 
on $I$ according to the CDF of the the distribution $\vec{p}$:
\begin{align}\label{eq:minp}
\vec{p} = \argmin_{\vec{p_s},\vec{p_d}}(H(\vec{p_s}), H(\vec{p_d}))
\end{align}
Let $U_I$ be a uniform random variable on $I$. The $i$-th point $x_i$ is then given by 
\begin{align}
	x_i = U_I + F_{i-1}  \quad \text{with} \quad i = 1, 2, \ldots, n.
\end{align}
Though adding $U_I$ is not crucial to our construction, 
the resulting randomness aids us in overcoming an adversary 
(see Section~\ref{sec:properties}). Node $u_i$ is therefore responsible for segment 
\begin{align}
	s_i = [(U_I + F_{i-1})~\text{mod}~1, \quad (U_I + F_{i})~\text{mod}~1),
\end{align}
of length $p_i$. For simplicity of presentation, in the
 following we omit the modulo operator, and all points $x \in I$ 
 are taken modulo 1. The rest of the discretization is like
 in the continuous-discrete approach. 
In the discrete graph $G_{\vec{x}}$, each \svr $\svi{i}$ is 
associated with the segment $s_i$, and we may refer to this segment 
as $\sv{i}$ as well. If a point $y$ is in $s_i$, we say that 
$\svi{i}\ covers\ y$. A pair of vertices $\svi{i}$ and  $\svi{j}$ has an edge 
$(\svi{i}, \svi{j})$ in $G_\vec{x}$ if there exists an edge $(y, z)$
 in the \emph{continuous} graph, such that $y \in s_i$ and $z \in s_j$. 
 The edges $(\svi{i}, \svi{i+1})$ and $(\svi{n - 1}, \svi{0})$ are added such that $G_\vec{x}$ contains a ring. 

An important feature of our design (which we use later in our routing algorithm) is the 
relationship between 
the segment  $\sv{i}$ of a node to its codeword $\cw{i}$ .
%
Let the \emph{ID} of \svr $u_i$ be $\cw{i}$:
\begin{align}\label{eq:cwvi}
\cw{i} = \lfloor (\bar{F}_i + U_I)_{01} \rfloor_{\ell_i},
\end{align}
Recall that $\ell_i$ is the length of $\cw{i}$, see Eq.\eqref{eq:cwlength}.
Define $\cs{i}$ to be the \emph{code segment} of $v_i$:
\begin{align}
\cs{i} = [\cw{i}, \cw{i}+2^{-\ell_i}) 
\end{align}

We will use the following property of $\cs{i}$:
\begin{property}\label{prp:prefix}
$\cs{i}$ contains all the points $z \in I$ s.t.~$\cw{i}$ is a prefix of $z$. 
\end{property}

\begin{clm} \label{clm:cwcs} 
$\cw{i} \in  \cs{i} \subseteq \sv{i}$
\end{clm}

We omit the proof here and note that by definition $\cw{i} \in  \cs{i}$, and
 $\cs{i} \subseteq \sv{i}$ is a basic property of Shannon-Fano-Elias coding, 
 so it is a prefix code.
Figure \ref{fig:Cont_Disc} illustrates 
the \emph{segment} and \emph{code segment} of $v_i$, for simplicity  $U_I=0$.

Let us clarify our approach with an example. Consider the activity distribution given in Table~\ref{tab:examplenetwork} and, for simplicity of presentation, set $U_I = 0$. 
Carrying out the codeword construction as shown in the table, we get the node placement of Figure~\ref{fig:example_cont}. To obtain the discrete graph one then checks how the left and right images of each segment intersect other segments. For a segment $s_i = [x_i, x_{i+1})$ its left  and right segments are,
$\ledge(s_i) = [\ledge(x_i), \ledge(x_{i+1}))$ and $\redge(s_i) = [\redge(x_i), \redge(x_{i+1}))$, respectively.
For instance, the left edges of segment $s_4$ partially cover $s_2$ and $s_3$. Therefore, in $G$ the neighbors of node $u_4$ are $u_2$ and $u_3$. In the same way, the right edges of $u_4$ partially cover $s_5$ and $s_6$, which makes the respective nodes neighbors of $u_4$. Repeating the same process for all nodes, we obtain the discrete graph $G_\vec{x}$ which is shown in Figure~\ref{fig:example_disc}.  Notice that $u_4$ has also an edge to $u_1$: this edge is a result of a left edge of $u_1$, or equally, a backward edge for a point in $s_4$.

\setlength{\extrarowheight}{2.5pt}
\begin{table*}[t]
\begin{center}
     \begin{tabular}{l l l l l l l l } 
         $i$ & $p_i$ & $F_i$ & $\bar{F}_i$ & $(\bar{F}_i)_{01}$& $\ell_i$ & $cw_i$ & $x_i$ \\ [0.5ex] 
         \hline  
         1 & 0.1 & 0.1 & 0.05 & 0.000011... & 5 & 00001 & 0\\ 
         2 & 0.15 & 0.25 & 0.175 & 0.001011... & 4 & 0010 & 0.1\\ 
         3 & 0.2 & 0.45 & 0.35 & 0.010110... & 4 & 0101 & 0.25\\ 
         4 & 0.25 & 0.7 & 0.575 & 0.100100... & 3 & 100 & 0.45\\ 
         5 & 0.1 & 0.8 & 0.75 & 0.110000... & 5 & 11000 & 0.7\\
         6 & 0.2 & 1.0 & 0.9 & 0.111001... & 4 & 1110 & 0.8 \\  [0.5ex] 
         \hline
    \end{tabular}
\end{center}
\caption{An example network with 6 nodes. For each node $i$, the table presents the node activity level $p_i$, the CDF $F_i$, the altered CDF $\bar{F}_i$ and its full binary representation $(\bar{F}_i)_{01}$, the code length $\ell_i$, the Shannon-Fano-Elias binary code $cw_i$, and the position of $x_i$ in $I$.}
\label{tab:examplenetwork}
\end{table*}

\begin{figure*}[t!]
 \centering
	\begin{subfigure}[b]{0.25\linewidth}
	\centering
	\includegraphics[width=\linewidth]{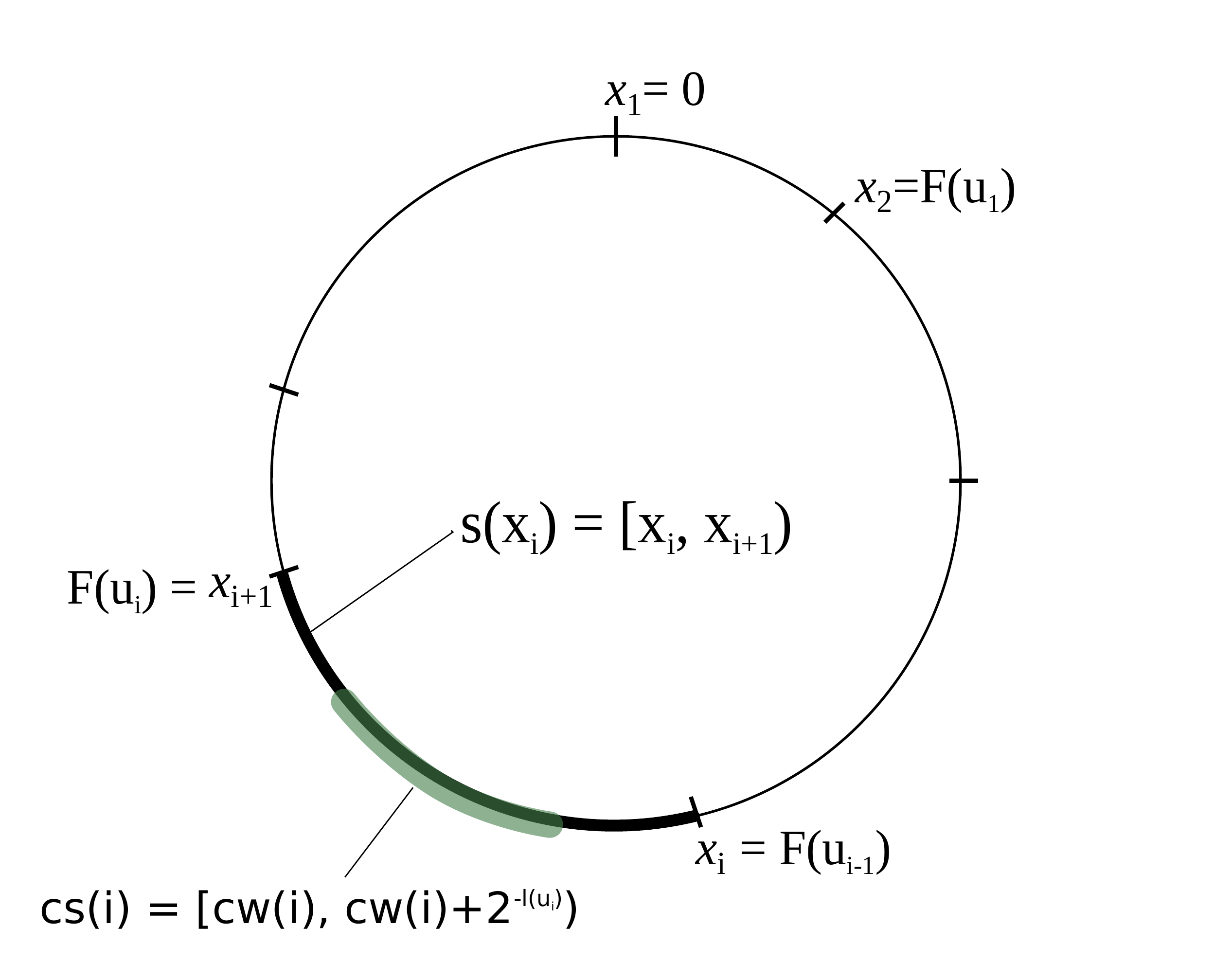}
	\caption{Segments and code segments}
	\label{fig:Cont_Disc}
	\end{subfigure}
    \centering
    \begin{subfigure}[b]{0.25\linewidth}
        \centering
        \includegraphics[width=\linewidth]{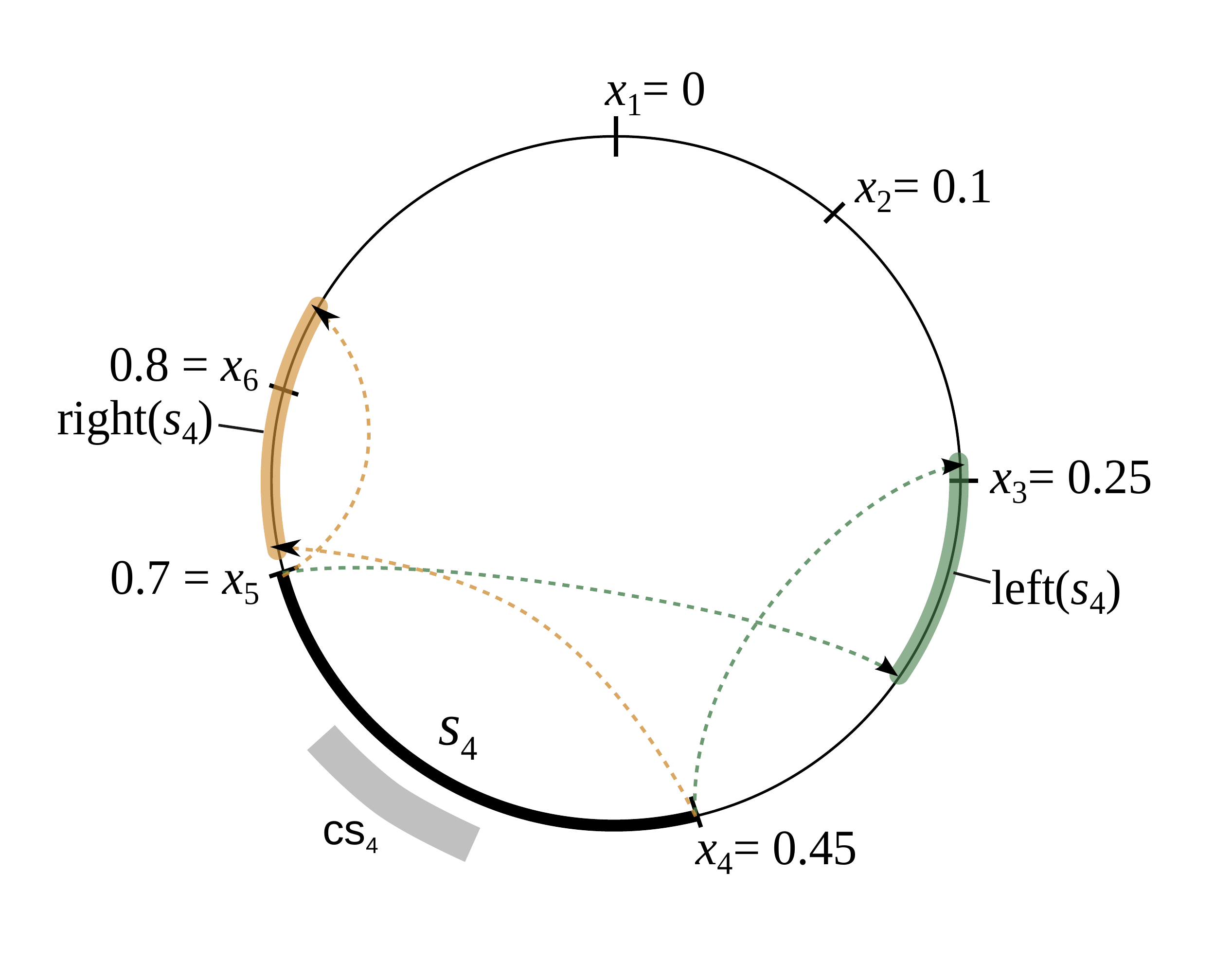}
        \caption{Left and right images of $s_4$}
        \label{fig:example_cont}
    \end{subfigure}
    \begin{subfigure}[b]{0.25\linewidth}
        \centering
        \includegraphics[width=\linewidth]{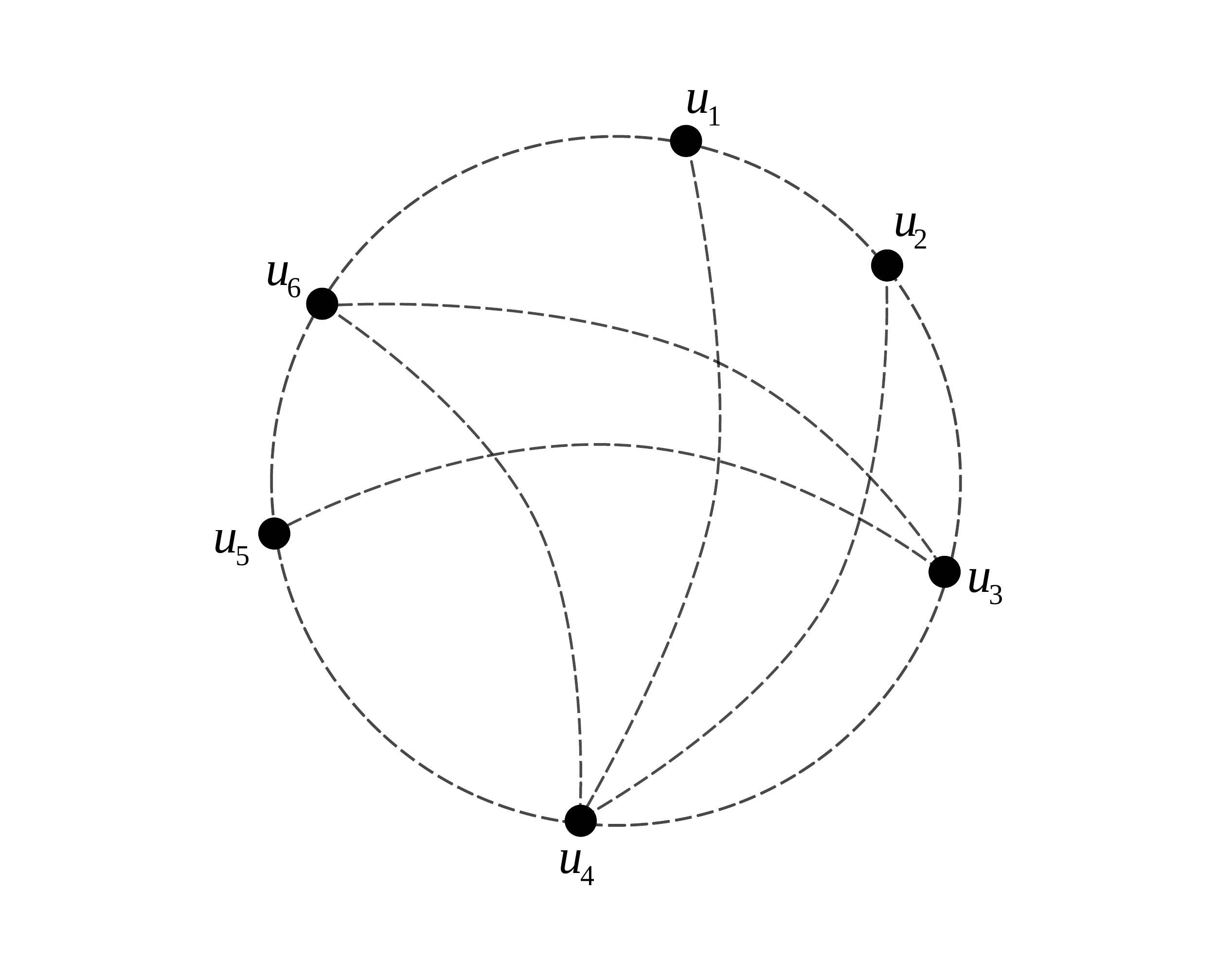}
        \caption{The resulting discrete graph $G_\vec{x}$.}
        \label{fig:example_disc}
    \end{subfigure}
    \caption{Illustration of  \System construction.}
    \label{fig:sysconstr}
\end{figure*}

\vspace{-2mm}

\section{\System Routing}\label{sec:algorithms}
\vspace{-1mm}

We investigate the routing efficiency for networks produced with our approach.

\paragraph{Routing Algorithms.}
In the spirit on the original continuous-discrete 
approach~\cite{naor2007novel} and de Bruijn graphs~\cite{bruijn1946combinatorial}, 
greedy routing can be done using two basic methods, 
\emph{forward} and \emph{backward} routing. 
Both methods were previously used for fixed length addresses, 
and thus in our construction require some adjustments due to the variable length of the node IDs.
We start with \emph{\textbf{forward routing}}.
Let $\sigi{i}=\cw{i}$ be the binary code and the ID for $u_i$ 
and let $\sigi{i}_t$ denote its suffix of length $t$.
Denote by $0.\sigi{i}_t$ and $1.\sigi{i}_t$ the concatenation 
of a bit to the string $\sigi{i}_t$. For every point $y \in I$ and 
for every \svr $u_i$, we define the function $\w{\sigi{i}_t, y}$ 
in the following recursive manner:

%
\begin{subequations}
    \begin{align}
        & \w{\sigi{i}_0, y} = y \\
        & \w{0.\sigi{i}_t, y} = \ledge{\w{\sigi{i}_t, y}} \\
        & \w{1.\sigi{i}_t, y} = \redge{\w{\sigi{i}_t, y}}
    \end{align}
\end{subequations}
In other words, $\w{\sigi{i}_t, y}$ is the point reached by a 
walk that starts at $y$ and proceeds according to $\sigi{i}_t$ 
from the least significant bit to the most significant bit of $\sigi{i}_t$.  

When \svr $u_i$ initiates a routing to the \svr $u_j$, 
the header of the message $u_i$ sends should contain the following information:
The destination ID $\sigi{j}=\cw{j}$, the source ID $\cw{i}$ (which is $u_i$s location), a counter $t$ initially set to 0, 
and a routing mode flag $f$, which defines if we use the forward ($f=1$) or backward ($f=0$) routing method.

%
The starting point of routing is at the source $u_i$ and with $t=0$ and $\w{\sigi{j}_0, \cw{i}}$. Upon receiving a message a node $u_{k}$ does the following:

\begin{algorithm}[H]
    \caption{Forward Routing - at node $u_{k}$}
        \begin{algorithmic}[1]
      		\State \textbf{if} $u_{k}$ is the destination: \textbf{done}.
            	\State \textbf{find} the node  $\svi{\mathrm{next}}$ 
            	which covers $\w{\sigi{j}_{t+1}, \cw{i}}$ 
            	\State \textbf{increase} $t$  and \textbf{forward} the message to  $\sv{\mathrm{next}}$
    \end{algorithmic}
        \label{alg:forward}
\end{algorithm}



\begin{lemma} \label{lem:rplf}
    For any two \svrs, a source $u_i$ and a destination $u_j$, the \emph{forward routing} will always reach the destination \svr and the routing path length will be at most $\ell_j$ hops.
\end{lemma}
\begin{proof}
    We first claim that routing on $G_c$ will reach $\sv{j}$ in $\ell_j$ hops.
     The routing starts at $\cw{i}$. Every hop effectively inserts 0 or 1 
     into the left of the current location, which defines a 
     new location on $I$. By the definition of {Forward Routing}, 
     every hop is done according to the appropriate bit of $\cw{j}$ 
     starting from the least significant bit. So the final location $z$ will be 
     the concatenation of the $\cw{j}$ and $\cw{i}$ codewords, 
     $z = \concat{\cw{j}}{\cw{i}}$ (where $\oplus$ is the 
     concatenation operator of two strings). 
     By Property~\ref{prp:prefix} and 
     Claim~\ref{clm:cwcs}, 
     $z=\concat{\cw{j}}{\cw{i}} \in \cs{j} \subseteq \sv{j}$, and therefore 
     it is covered by \svr $u_j$. 
	The $\cw{j}$ bit string length is equal to $\ell_j$, 
	therefore the routing path length in  $G_c$ will be at most $\ell_j$ hops. 
 To conclude the proof we note that every hop in the 
 continuous graph between $x$ and $y$ is possible in the 
 discrete graph since if $x \in \sv{k}$ and $y \in \sv{k'}$ 
 there must be an edge (by definition) between $u_{k}$ and $u_{k'}$.
\end{proof}

Let us provide a small forward routing example from the network of Table \ref{tab:examplenetwork} and Figure \ref{fig:example_disc}.
Consider \svr $u_{6}$ (ID $1110$) to be the source and \svr $u_{4}$  (ID $100$) to be the destination.
In the continuous graph the route starts at \svr $u_{6}$ at the point $0.1110$, 
and forwarding is according to the destination address bits (from least to most significant),
to the location $0.{\bf 0}1110$, then $0.{\bf 00}1110$, and last $0.{\bf 100}1110$.
In the discrete  graph, the message is sent from $u_{6}$ to $u_{3}$ to $u_{2}$ to $u_{4}$ and has 3 hops, the length of the destination ID.
Note that the last point $0.{\bf 100}1110$ is guaranteed to belong to the \emph{code segment} of  $u_{4}$ since its prefix is $100$. 


We now consider the second routing algorithm, called \emph{\textbf{backward routing}}.
Assume \svr $u_i$ wants to route to \svr $u_j$. 
Now, a message traveling 
from $u_i$ to $u_j$ would start at the point $\concat{\sigi{i}}{\sigi{j}} 
\in \cs{i}$ and travel the path backwards, along the backward edges, 
until it has reached $\cs{j}$. 
The source \svr $u_i$ creates messages that contain the 
source $ID$ equal to $\cw{j}$, target $ID$ $\cw{j}$, 
counter $t$ initially set to 0 and routing mode flag $f=0$, 
which defines the backward routing method. 
The following is a more formal description of the protocol at node $u_{k}$:

\begin{algorithm}[H]
    \caption{Backward Routing - at node $u_{k}$}
        \begin{algorithmic}[1]
      		\State \textbf{if} $u_{k}$ is the destination: \textbf{done}.
            	\State \textbf{find} the node  $\svi{\mathrm{next}}$ which covers $\concat{\sigi{i}_{l(i) - t -1}}{\cw{j}}$ 
            	\State \textbf{increase} $t$  and \textbf{forward} the message to  $\svi{\mathrm{next}}$
    \end{algorithmic}
    \label{alg:backward}
\end{algorithm}


Similarly to forward routing we can bound 
the hops:

%

\begin{lemma} \label{lem:rplb}
    For any source $u_i$ destination $u_j$ pair, 
    the \emph{backward routing} will always reach the destination \svr, and the routing path length is at most $\ell(i)$ hops.
\end{lemma}

We can now bound the expected path length.

\paragraph{Expected Path Length.}
We proceed to characterize the expected routing length for a request distribution matrix $\mat{R}$ with source and destination marginal distributions $\vec{p_s}$ and $\vec{p_d}$ and corresponding entropy measures of $\vec{p_s}$ and $\vec{p_d}$, respectively. 
%

In particular, we are interested in algorithm  $\mathcal{A}$ that uses \emph{forward routing} whenever $H(\vec{p_s}) \geq H(\vec{p_d})$ and \emph{backward routing} when $H(\vec{p_s}) \leq H(\vec{p_d})$. The following theorem bounds the expected path length of $\mathcal{A}$.

\begin{theorem} \label{th:erplb}
    For any request distribution $\mathcal{R}$, the expected path length  satisfies
$	      \textrm{EPL}(\mathcal{R}, G, \mathcal{A}) < \min\{H(\vec{p_s}), H(\vec{p_d})\} + 2.
$
\end{theorem}
\begin{proof}
    By Lemmas \ref{lem:rplf} and \ref{lem:rplb}, 
    for any two \svrs source $u_i$ and destination $u_j$, 
    the routing path length is at most the codeword length. 
    By  Eq.~\eqref{eq:minp} $\vec{p}$ is the marginal distribution with minimum entropy
    and the distribution by which we build the network. 
    Based on Eq.~\eqref{eq:ecl} and Eq.~\eqref{eq:eclbound}, 
    the expected routing path length is
    \begin{align*}    
         \textrm{EPL}(\mathcal{R}, G, \mathcal{A})  &= 
         \sum\limits_{u_i,u_j \in V} R_{ij} \cdot Route_{G, \mathcal{A}}(i,j) \nonumber \\
           &\le \sum\limits_{u_j \in V}\sum\limits_{u_i \in V} R_{ij} \cdot \ell(j) 
           \\ & =  \sum\limits_{u_j \in V}  \ell(j) \sum\limits_{u_i \in V} R_{ij} \nonumber \\
           &=  \sum\limits_{u_j \in V} p_j \cdot \ell(j) = \sum\limits_{u_j \in V} p_j(\lceil \log{\frac{1}{p_j}} \rceil + 1)\nonumber \\
            &< \sum\limits_{u_j \in V} p_j(\log{\frac{1}{p_j}} + 2) = H(\vec{p}) + 2 
    \end{align*}
\end{proof}


A nice observation is that we can propose an \emph{\textbf{Improved Routing Algorithm}}
by combining \emph{forward routing} and \emph{backward routing}. 
Each node that initiates a routing decides on the routing mode. 
If the destination \svr codeword is shorter than the source codeword, 
then it selects the \emph{forward routing} mode, 
otherwise it selects the \emph{backward routing} mode. 
A relay \svr process according to mode, defined by source \svr.
Let the improved algorithm be denoted by $\mathcal{A}^*$. 
\begin{clm} \label{clm:rpl}
    For any two \svrs source $u_i$ and destination $u_j$, 
    the routing path length using the improved 
    algorithm $\mathcal{A}^*$ will be $\min(\ell(i), \ell(j))$.
\end{clm}
%
In other words, combining forward and backward routing can only help and $\textrm{EPL}(\mathcal{R}, G, \mathcal{A}^\ast) \le \textrm{EPL}(\mathcal{R}, G, \mathcal{A})$.



\paragraph{Routing Under Failure.}
In case of edge failures, our routing algorithms could be easily resumed 
by sending the message to \emph{any} available neighbor.
We add this feature to our algorithms, and every time when a next hop 
edge fails at \svr $u_i$, we select uniformly at random any valid edge to 
a neighbor
 \svr $u_j$, reset the routing message and send it to $u_j$, 
 to continue routing as if it is a new route starting from  \svr $u_j$. 
 Our construction contains rings, therefore to prevent infinite loops, 
 we will define the maximum routing length restrictions (TTL).
 
Next we prove other important network properties like sparsity, fairness and robustness.


\vspace{-1mm}

\section{\System Network Properties}
\label{sec:properties}

Let us take a closer look at the basic connectivity properties of the networks designed by our approach. We start with deriving basic connectivity properties, especially regarding sparsity, degree fairness, and robustness. Our results demonstrate that, even though the designed networks have low degree in expectation, each network is well-connected.
\paragraph{Sparsity.}
Similar to the original Continuous-Discrete approach, we can prove that our network is sparse.
The proof is essentially the same as in the original construction and we omit it here.
\begin{proposition} \label{th:edgenum}
    The total number of edges in
    $G$, without the ring edges, is at most $3n-1$.
\end{proposition}

\paragraph{Fairness and Nodes Degree.}
Our requirement for fairness (Section \ref{sec:model}) was that node degree will be proportional to its activity and
so low activity nodes will not have a high degree.
Using the original Continuous-Discrete approach we are able to extend their results and show an upper bound on a node degree that is based on
its activity. Recall that in our construction the activity level of a node $\svi{i}$ is equal to $p_i = \abs{\sv{i}}$ and denoted as $p_i$.

\begin{dfn}
    Let the length of the minimal segment of $\vec{x}$ be denoted
     by $p_{\min} = \min\limits_i p_i$, and let the corresponding node be denoted by
    $\svi{\min}$. For $\svi{i}$ let $\rho_i = p_i/p_{\min} = 
    \abs{\sv{i}}/\abs{\sv{\min}}$.
\end{dfn}

\begin{proposition}
	The maximum out-degree of $\svi{i}$ without the ring edges is at most $\rho_i + 4$,
	and the maximum in-degree without the ring edges is at most $\lceil2\rho_i\rceil + 1$.
\end{proposition}


Since our construction is randomized (based on the random shift $U_I$), we can provide a tighter bound on the expected degree of a node  $\svi{i}$, and show that
the expected out-degree and in-degree of each node $u_i$ is proportional to $n \cdot p_i$. 
\begin{lemma}
	The expected out-degree of node $u_{i}$ is $1+ 0.5 (n-3) \cdot p_i$
	and the expected in-degree is $0.5+  (n-1.5) \cdot p_i$.
	\label{prop:degree}
\end{lemma}
\begin{proof}
From the linearity of expectation, we have

{\footnotesize
\begin{align}
	\E{d_i} &= \E{ \sum_{j \neq i} L_{ij} + R_{ij}} =  \sum_{j \neq i} \P{L_{ij} = 1 \vee R_{ij} = 1}
	\label{eq:expectation}
\end{align}   
}

where each indicator random variable $L_{ij}$ and $R_{ij}$ becomes $1$ if the segment of node $u_j$ intersects that of the \emph{left} and \emph{right} images of $u_{i}$, respectively. 

Consider first the right-image intersection 
random variable $R_{ij}$, and denote by $\delta = x_j - x_i$ the distance between the start of segments $s_j$ and $s_i$. Since node $u_j$ can only be a right-neighbor of $u_i$ if segments $\redge{s_i}$ and $s_j$ intersect, we have that
\begin{align*}
\P{R_{ij} = 1} &= \P{	x_j - |\redge{s_i}| \leq \redge{x_i} \leq x_j + |s_j|}  \\
			&= \P{ x_i + \delta - \frac{p_i}{2} \leq \frac{x_i + 1}{2} \leq x_i + \delta + p_j} \\
			&= \P{ \delta - \frac{p_i}{2} \leq U_{[0,0.5)} \leq \delta + p_j }		
\end{align*}
where $U_{[0,0.5)} = \frac{1 - x_i}{2}$ is a random variable that lies uniformly in $[0, 0.5)$.

Similarly, for the left image
\begin{align*}
\P{L_{ij} = 1} &= \P{x_j - |\ledge{s_i}| \leq \ledge{x_i} \leq x_j + |s_j|}  \\
			&= \P{ x_i + \delta - \frac{p_i}{2} \leq \frac{x_i}{2} \leq x_i + \delta + p_j} \\
			&= \P{ \delta - \frac{p_i}{2} \leq U_{[0.5,1)} \leq \delta + p_j }		
\end{align*}
Separately, the intersection probabilities depend on $\delta$ (for instance, when $\delta > 0.5 (1+ p_i)$ the right intersection probability is always zero). The probability of the joint event however is independent of $\delta$.
\begin{align*}
\P{L_{ij} = 1 \vee R_{ij} = 1} &= \P{ \delta - \frac{p_i}{2} \leq U_{I} \leq \delta + p_j } \\
&= p_j + 0.5\,p_i		
\end{align*}
Substituting to Equation~\eqref{eq:expectation}, we get that the expected out-degree of $\svi{i}$ is:
\begin{align*}
	 \sum_{j \neq i} \left( p_j + 0.5\, p_i\right) = 0.5\, (n-3) p_i + 1, 
\end{align*}
and our claim follows.
Very similarly we can show the expected in-degree.
\end{proof}

\paragraph{Connectivity and Robustness.}
Our goal is to show that the \System topology is highly connected and robust to edge failures. In particular, we claim that in order to disconnect 
a set with high activity, one needs to fail many edges.
More formally for a set $S \subseteq V$, let $p_S$ denote its probability according to $\vec{p}$, i.e., $p_S=\sum_{\svi{i} \in S} p_i$.
The \emph{cut} $C(S, \bar{S})$ is the set of edges connecting $S$ to its compliment, $\bar{S}=V \setminus S$:  $C(S, \bar{S}) = \{(u,v) : u \in S, v \in \bar{S}, (u,v) \in E\}$.
We can claim the following. 

\begin{theorem} \label{th:cut}
    For any set of nodes $S \subseteq V$ s.t.~$p_S \le 1/2$, its expected number of edges in the cut is:
$	    \E{\abs{C(S,\bar{S})}} \ge \Omega(\frac{n p_S}{\log p_{\min}})
$%
\end{theorem}

For example this theorem entails that if $\log p_{\min} = \Omega(\log n)$ then disconnecting a set with constant activity $p_S$ will require 
failing many edges, namely $\Omega(n/\log n)$ edges. Recall that in a tree network there are sets where this can be achieved by failing exactly one edge.

\begin{proof}
To prove this we will use the \emph{expansion} properties of de Bruijn graphs.
The \emph{edge expansion}~\cite{hoory2006expander} of a graph $G$ is defined as:
\begin{align}
h(G) = \min_{0 < \abs{S} \le \frac{n}{2}} \frac{E(S,\bar{S})}{\abs{S}}
\end{align}
Then for a graph with expansion $\alpha$ and a set $S$ (assume \wlg that $\abs{S} \le \abs{\bar{S}}$), the number of edges in the cut is at least 
$\abs{E(S,\bar{S})} \ge \alpha \cdot \abs{S}$. 
It is known that the expansion of a de Bruijn graph with $2^r$ nodes is $\Theta(1/r)$~\cite{Leighton:1991:IPA:119339}. Our first step will be to bound the image of $S$ in the continues graph.
Let $Im(S)$ denote the set of points $x \in I$ s.t. $x$ has a neighbour in $S$ in the continues graph $G_c$. We can claim:
\begin{clm}
\begin{align}\label{eq:img}
\abs{Im(S)} = \Theta(\frac{p_s}{\log p_{\min}})
\end{align}
\end{clm}
\begin{proof}
Let $r= \lfloor \log p_{\min}/3 \rfloor$ and recall that if we discretize the continuos graph into uniform size segments of size $2^{-r}$, it results in a de Bruijn graph with $2^r$ nodes~\cite{naor2007novel}. Denote this graph by $G_r$. Since the expansion of $G_r$ is $\Theta(\frac{1}{r})$ and $S$ has $\Theta(\frac{p_S}{2^r})$ nodes in $G_r$,
the edge cut size, $\abs{C(S,\bar{S})}=\Theta(\frac{p_s 2^r}{r})$ in $G_r$. Since $G_r$'s maximum degree is 4 and the length of each segment is $2^{-r}$,
the claim follows.
\end{proof}
To bound $\E{\abs{C(S,\bar{S})}}$ we need now to bound the expected number of segments in $\bar{S}$ that intersect with $Im(S)$.
This could be done with indicator functions similar to the proof which bound the expected degree.
Recall that $p_s \le 1/2$ and we additionally assume \wlg that $\abs{S} \le n/2$ (if this is not the case, we can replace $S$ with $\bar{S}$ to our benefit).
$Im(S)$ may be a union of disjoint segment. Let $s'_1, s'_2, \dots s'_k$ denote these segments s.t. $Im(S) = \cup s'_i$ and $\abs{Im(S)} = \sum \abs{s'_i}$.
Assume $\bar{S}$ contains $\ell> n/2$ nodes with corresponding segments $v'_1, v'_2, \dots v'_\ell$. Let the indicator function $I_{i,j}$ denote if
segment $s'_i$ intersects with segment $u'_j$. Note that in this case node $u'_j \in \bar{S}$ will have an edge with a node in $S$.
We can now bound $\E{\abs{C(S,\bar{S})}}$ as 
\begin{align}
\E{\abs{C(S,\bar{S})}} = \E{\sum_{i,j} I_{i,j}} = \sum_{i,j} \E{I_{i,j}} 
\end{align}
Since $v'_j$ is uniformly distributed in $I$ we have
\begin{align}
\E{\abs{C(S,\bar{S})}} &= \sum_{i,j} \abs{s(u'_j)} + \abs{s'_i} \\ \notag
&=   \sum_{i=1}^k \sum_{j=1}^\ell \abs{s(u'_j)} +  \sum_{j=1}^\ell  \sum_{i=1}^k \abs{s'_i} \\ \notag
&\ge  \sum_{j=1}^\ell \abs{Im(S)}  \ge \frac{n}{2} \abs{Im(S)} 
\end{align}
This together with Eq.~\refeq{eq:img} concludes the proof.
 \end{proof}

In addition we can prove that $G_\vec{x}$ is 2-edge-connected so no removal of a single edge can disconnect the graph.
\section{Experiments}
\label{sec:sims}

We complement our formal analysis with a simulation study.

\paragraph{Setup and Methodology.}
The following experiments illustrate the properties which emerge when one enhances the continuous-discrete approach by taking the request distribution $\mat{R}$ into account.  
%
%
At the heart of our model lies the entropy of the marginal source and destination distributions $\vec{p_s}$ and $\vec{p_d}$. 
In order to test our construction across a range of possible entropies, 
we determined the per-node demand (the source or the destination) 
according to Zipf's law:
$	p_i =  i^{-s}/\sum_{j = 1}^n j^{-s}.
$
We aim to model a power law distribution where the parameter $s$  
controls the (possibly high) variance in the activity level of nodes, as well the entropy. 
It is known that:
\begin{align}
	H(\vec{p}) = \frac{s}{H_{n,s}} \sum_{i = 1}^n \frac{\log{i}}{i^s}  + \log{H_{n,s}},
	\label{eq:Zipf_entropy}
\end{align}
where the exponent $s \geq 0$ and $H_{n,s}$ is the $n$th generalized harmonic number. 
After constructing marginal distributions $\vec{p_s}$ and $\vec{p_d}$ according to Zipf's law (with different id permutations for each case) 
we construct the demand matrix $\mat{R}$ as the product distribution $\mat{R} = \vec{p_s} \vec{p_d}^\top$.

Admittedly, this is only one of many possible choices for the request distribution and, being a product distribution, our chosen $\mat{R}$ does not possess the strong locality properties which are sometimes present in real networks. We argue however that our choice is a natural one: Power law distributions have been observed and used in many contexts, and are well-studied. Moreover, by varying $s$, we can test our approach across the entire entropy spectrum, from the extremes of maximum entropy ($s=0$, uniform activity levels) to zero entropy ($s \gg 0$, very large variance). 
Notice that, even though this chosen demand distribution is deterministic, 
our construction is not---the ordering of the nodes on the ring can be arbitrary. 
To evaluate our approach across all possible orderings, we performed multiple iterations, 
permuting the node ordering randomly each time.

\paragraph{Fairness and Nodes Degree.}
We first investigate the properties of the network topology itself, 
regardless of the routing method. 
In particular, we would like to show that our design leverages 
the demand distribution to achieve a good trade-off between 
two competing objectives: \emph{(i)} Attaining a well connected 
network in expectation. This implies that paths between nodes that communicate often 
should be short. And \emph{(ii)} 
having a degree that is not only small, but also proportional to the demand. 

To this end, we generated networks of 300 \svrs, 
with two representative values of the Zipf's exponent, 
$s = 0.5$ and $s=1$ for the node activity distribution 
(larger values of $s$ result in an unrealistic distributions where 
more than 50\% of the probability mass is concentrated in a single node). 
Figure~\ref{fig:degree} depicts the node degree of each node 
$u_i$, averaged over 50 realizations, \emph{(vertical axis)}, as a function of its activity level $p_i$ \emph{(horizontal axis)}. 
As expected, in the original Continuous-Discrete 
approach (CD) shown in red there is no correlation between $p_i$ and node degree, 
as the construction ignores $\vec{p}$. In contrast, and similar 
to our analysis, in the proposed code-based approach (CB) shown in black node degrees are correlated with $p_i$ and proportional to 
$n p_i$. Nodes with low activity level will not have high degree, 
while nodes with high activity levels will have on average higher degrees. 
Note that this correlation is based on $\vec{p}$ which is the base for the 
construction. Since $\vec{p}$ can be, for example, the source 
distribution (and not the destination distribution), it could still be the case that 
a node is highly active as destination but has low degree. But as we require, the 
opposite will not happen, a node with low activity 
(either as source or destination) will not have a high degree.

\begin{figure*}[t]
	\centering
	\includegraphics[width=0.48\linewidth]{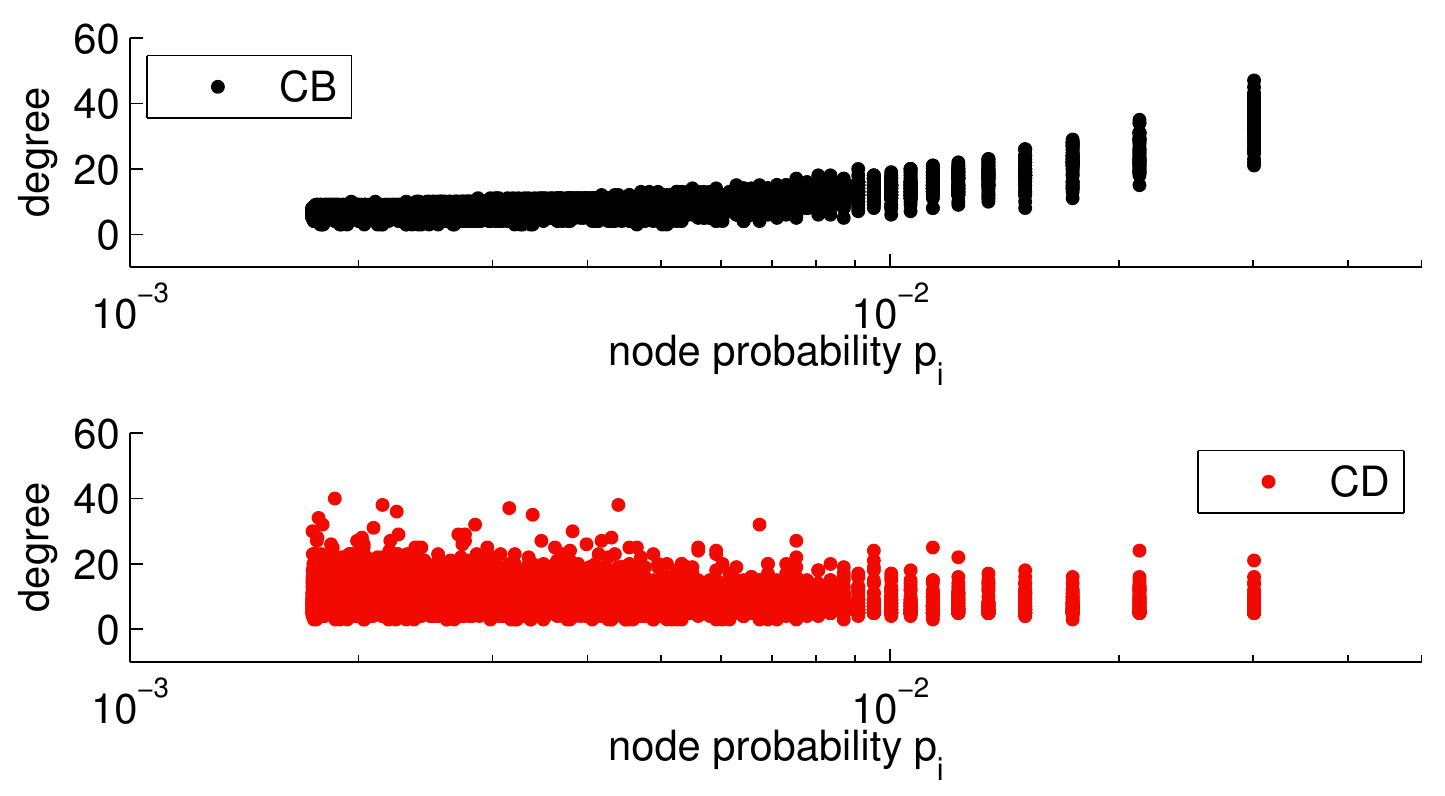}
	\includegraphics[width=0.48\linewidth]{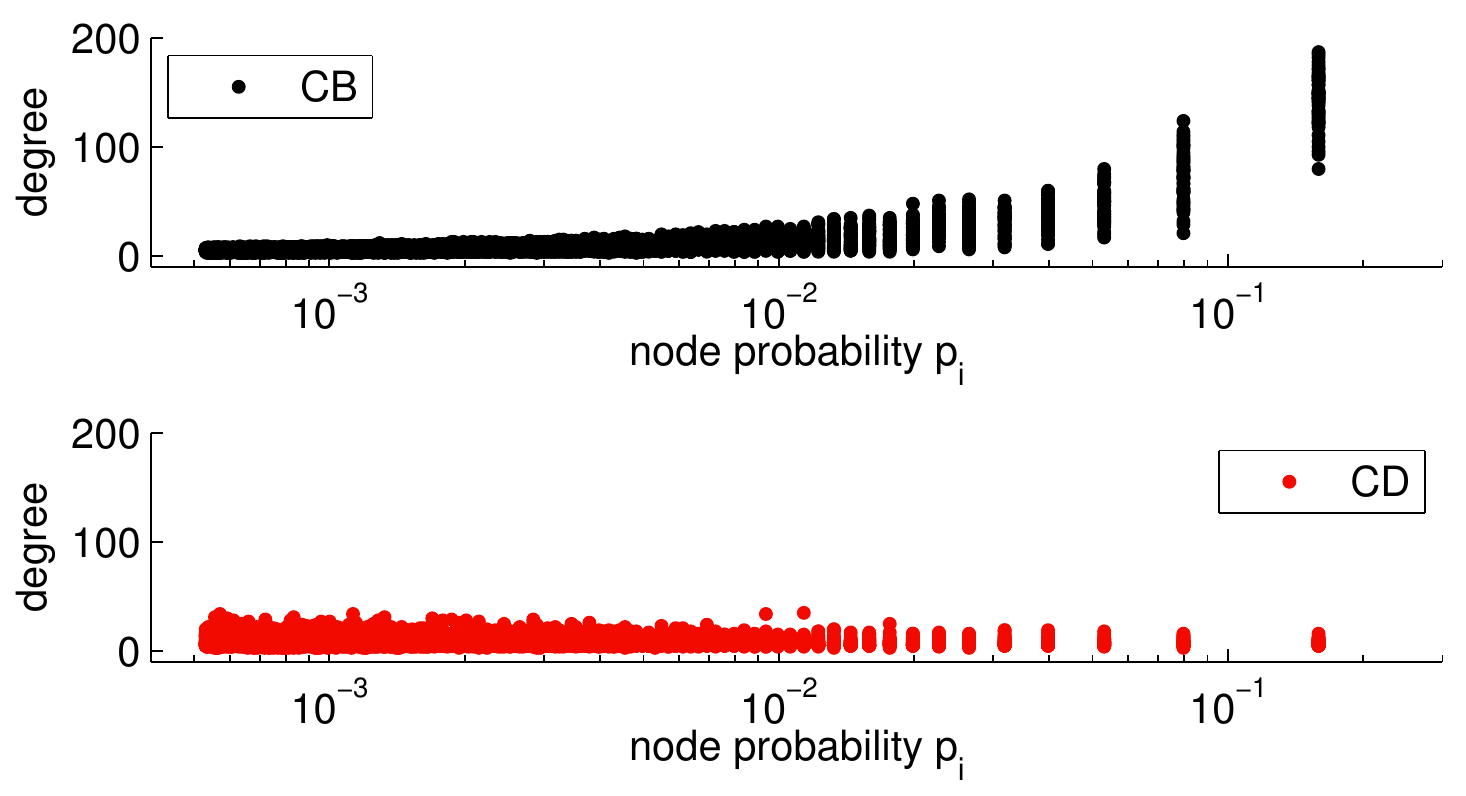}\\
    \caption{Node degree as a function of node probability. Code-based design ($CB$) demonstrate that low activity nodes have low degree and high activity nodes may have high degree.
\emph{Left} $s = 0.5$, \emph{right} $s = 1.0$.}
    \label{fig:degree}
\end{figure*}



\paragraph{Expected Path Length.}
We proceed to examine the efficiency of code-based routing (CBR) and to compare it to that in the original continuous-discrete approach (CDR). 
Formally, code-based routing (CBR) is the \emph{improved routing} based on Claim \ref{clm:rpl}, a routing that compares the source and destination IDs and
uses the shorter route based on forward or backward edges, Algorithms  \ref{alg:forward} and \ref{alg:backward}, respectively. We use CBR (fwd) to explicitly denote using forward routing on the \System network.
The original continuous-discrete approach (CDR) is using forward routing. 
\begin{figure}[ht]
    \centering
	\includegraphics[width=0.9\linewidth]{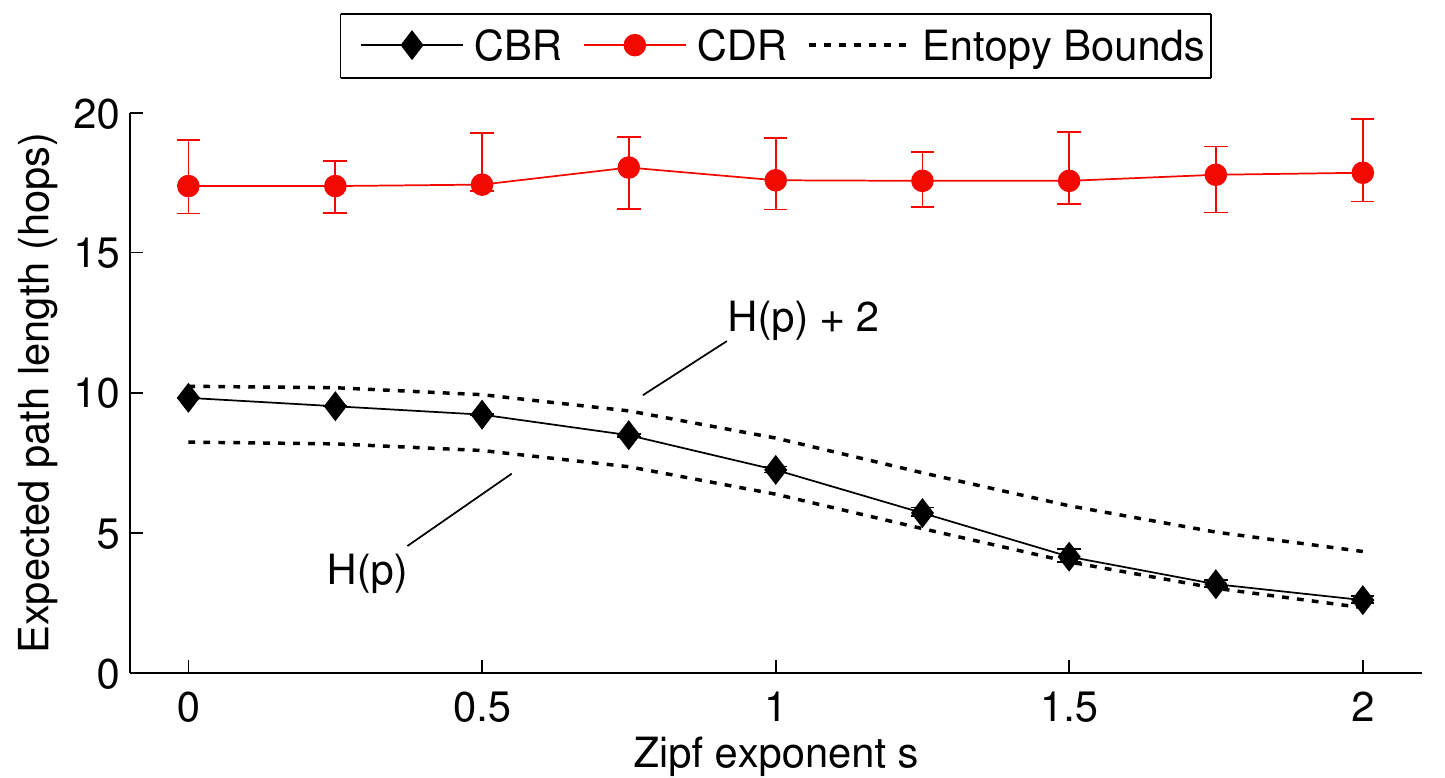}
    \caption{EPL as a function of Zipf's law exponent $s$ comparing to the Entropy.} 
    \label{fig:exp}
\end{figure}

\vspace{3mm}\emph{Insight 1. Entropy determines the efficiency of routing.} To test the role of entropy, we measure the expected path length (EPL) defined in~Eq.~\eqref{eq:EPL} for networks of 
size 300 over a wide range of Zipf's law exponents. For every exponent $s$ we generate $\vec{p_s}$ and $\vec{p_d}$, repeat 50 simulations, and take the average.  
As shown in the Figure~\ref{fig:exp}, by optimizing the topology design using the request distribution one obtains a dramatic improvement of routing path length (in expectation). The role of entropy as an upper bound is also clear from the figure. Whereas the performance of CDR remains constant across all exponents $s$ (we remind the reader that $s$ controls the entropy---see Eq.~\ref{eq:Zipf_entropy}), using CBR with highly skewed distributions (large $s$) results in shorter paths. Moreover, the EPL is lower and upper bounded by $H(\vec{p})$ and $H(\vec{p})+2$, respectively. 

\vspace{3mm}\emph{Insight 2. Frequently occurring paths are shorter.} The relation between path probability and path length is more clearly depicted in Figure~\ref{fig:pathlength}. In the figure, we choose exponent $s = 1$ for which $H(\vec{p})=6.37$ and plot for each node pair the path probability (horizontal axis) and the corresponding routing path length averaged across all iterations (vertical axis).
As per our analysis, in CBR, the path length between two nodes $u_i$ and $u_j$ depends on the minimum of their assigned codewords $\cw{i}$ and $\cw{j}$ (Lemma \ref{lem:rplf}), which in turn is inversely proportional to $\min\{ p_i, p_j \}$. Simply put, by design, the more active a node pair is, the shorter the path between them. 
In particular, while the longest path is as high as 12 hops, the overwhelming majority of frequently occurring paths are very short. Specifically, we can see that all paths with probability larger than $10^{-4}$ are of length below 8. On the other hand, the paths of CDR are approximately 18 hops long, and are independent of the path probability. 

\begin{figure}[t]
	\centering
	\includegraphics[width=0.9\linewidth]{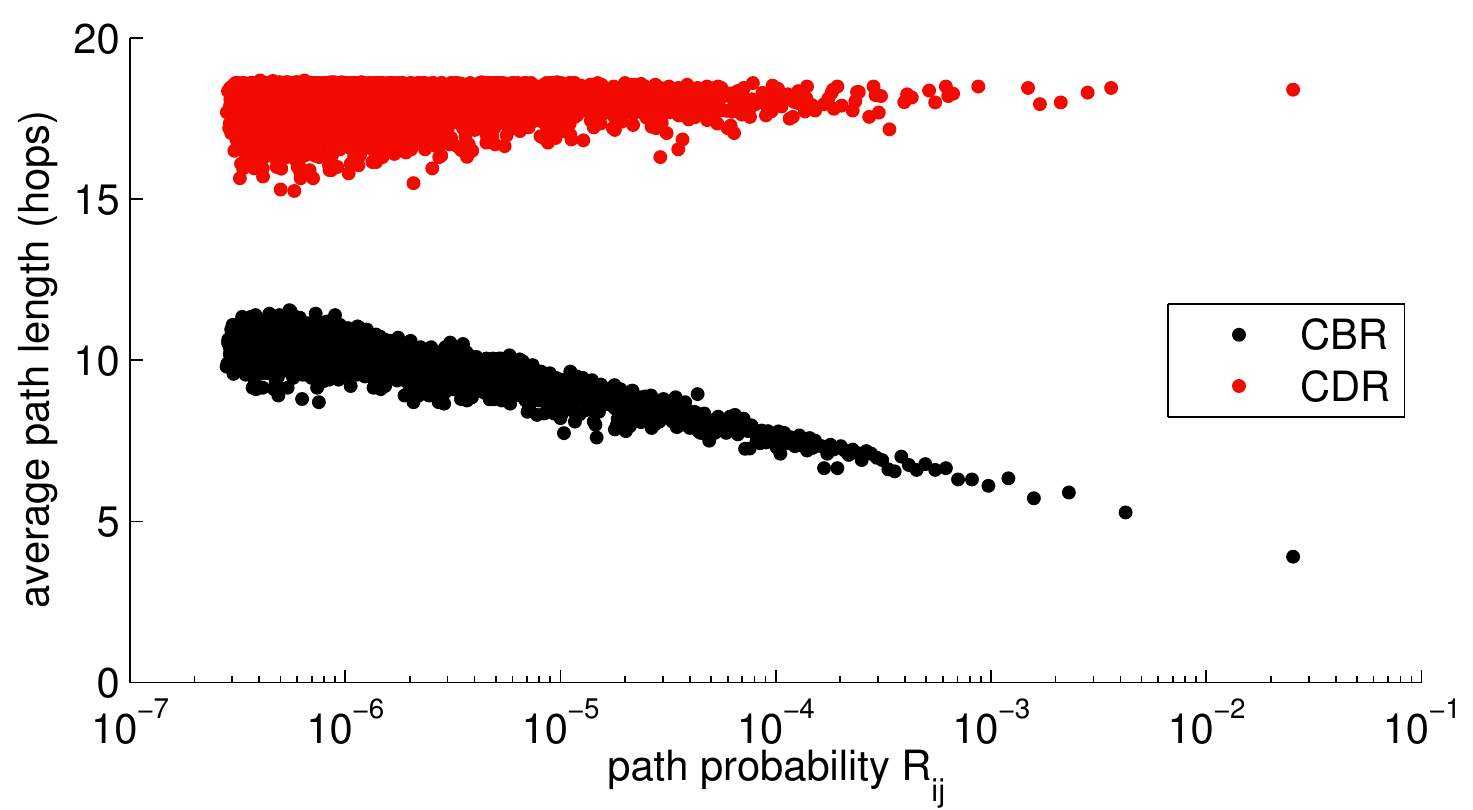}
	\caption{Average path length as a function of path probability $R_{ij}$. $CBR$ demonstrate that high probability routes have shorter length.}
    \label{fig:pathlength}
\end{figure}

\vspace{3mm}\noindent\emph{Insight 3. Scalability.} Importantly, how are our results affected by the network size? To find out, we set $s=1$ and repeated our experiments over networks with size varying from 100 to 1,000 nodes.
Figure~\ref{fig:netsize} shows the EPL for each network size, averaged over 10 iterations. It is clear from the figure, that for CBR the path length grows with $H(\vec{p})$, which can be much lower than the $O(\log{n})$, achieved by CDR.

\begin{figure}
\centering
        \includegraphics[width=0.9\linewidth]{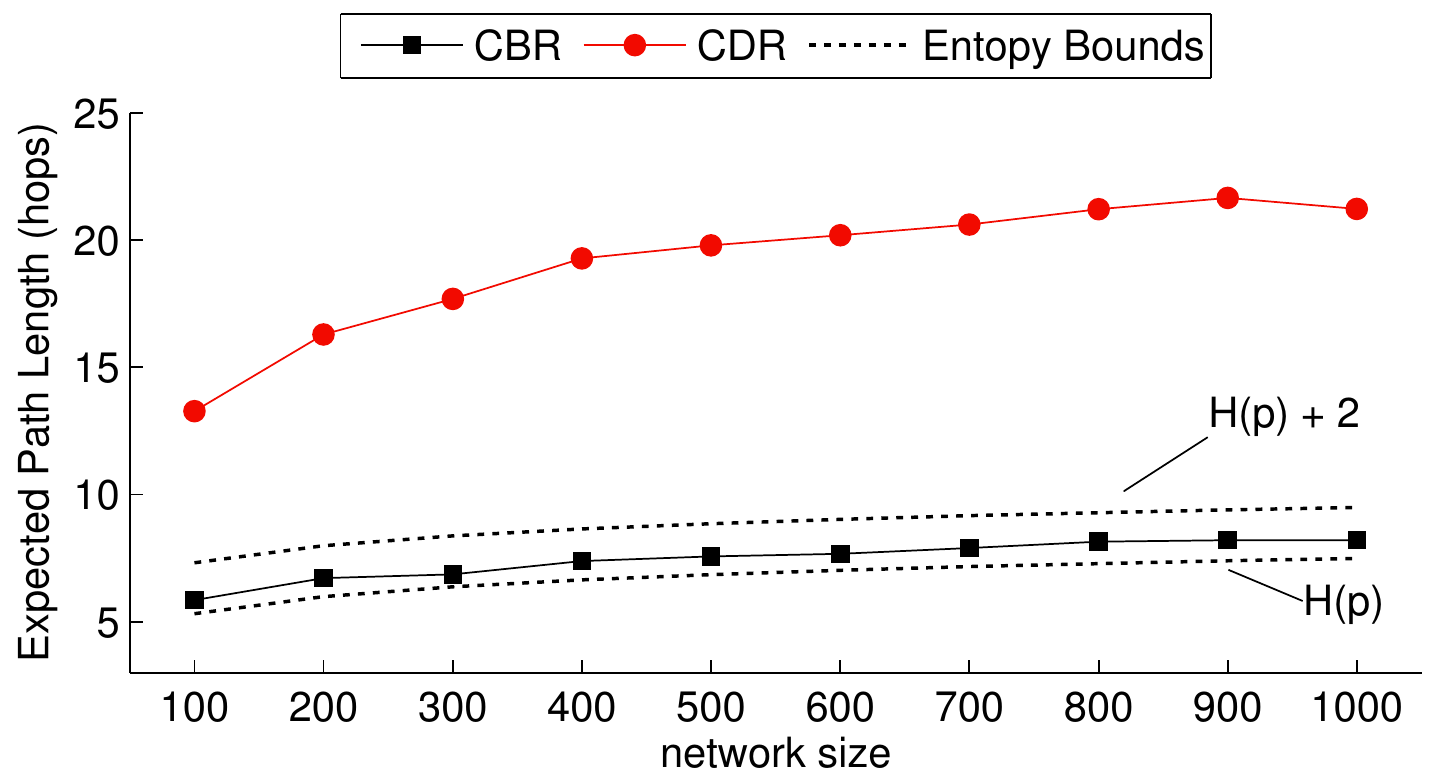}
    \caption{Expected path length as function of network size. $CBR$ grows with entropy.}
    \label{fig:netsize}    
\end{figure}

\paragraph{Load Balancing.}
Next we study load properties. 
We set to answer two questions: 
First, is there is a correlation between node activity and node
load? Recall that nodes also serve as 
\emph{relays} helping to forward messages from other nodes.
Second, is the load balanced also on the edges of the network? 

In Section~\ref{sec:properties}, we showed that every node $u_k$ has in expectation a degree proportional to its activity $a_k$. Our results however do not guarantee the same about the load of each node, meaning the aggregate traffic $u_k$ relays.
As it turns out, the answer is negative. Our simulations presented in 
Figure~\ref{fig:nodeload} show that low activity (and degree) 
nodes also do not serve much traffic.  
As before, the simulation were conducted 
on networks of 300 \svrs, where the results are the average 
of many simulations.
For each node $u_k$, we summed the probabilities  $R_{ij}$ of all routing paths which go via relay $u_k$, with $i,j \neq k$.  
Similar to the \svrs degree case, the node load in 
the Code-Based case is proportional to the demand. 
However, we can see that there is no dependency between \svr relay load 
and the activity probability in the Non Code-Based case.

\begin{figure}[t]
	\centering
	\includegraphics[width=0.9\linewidth]{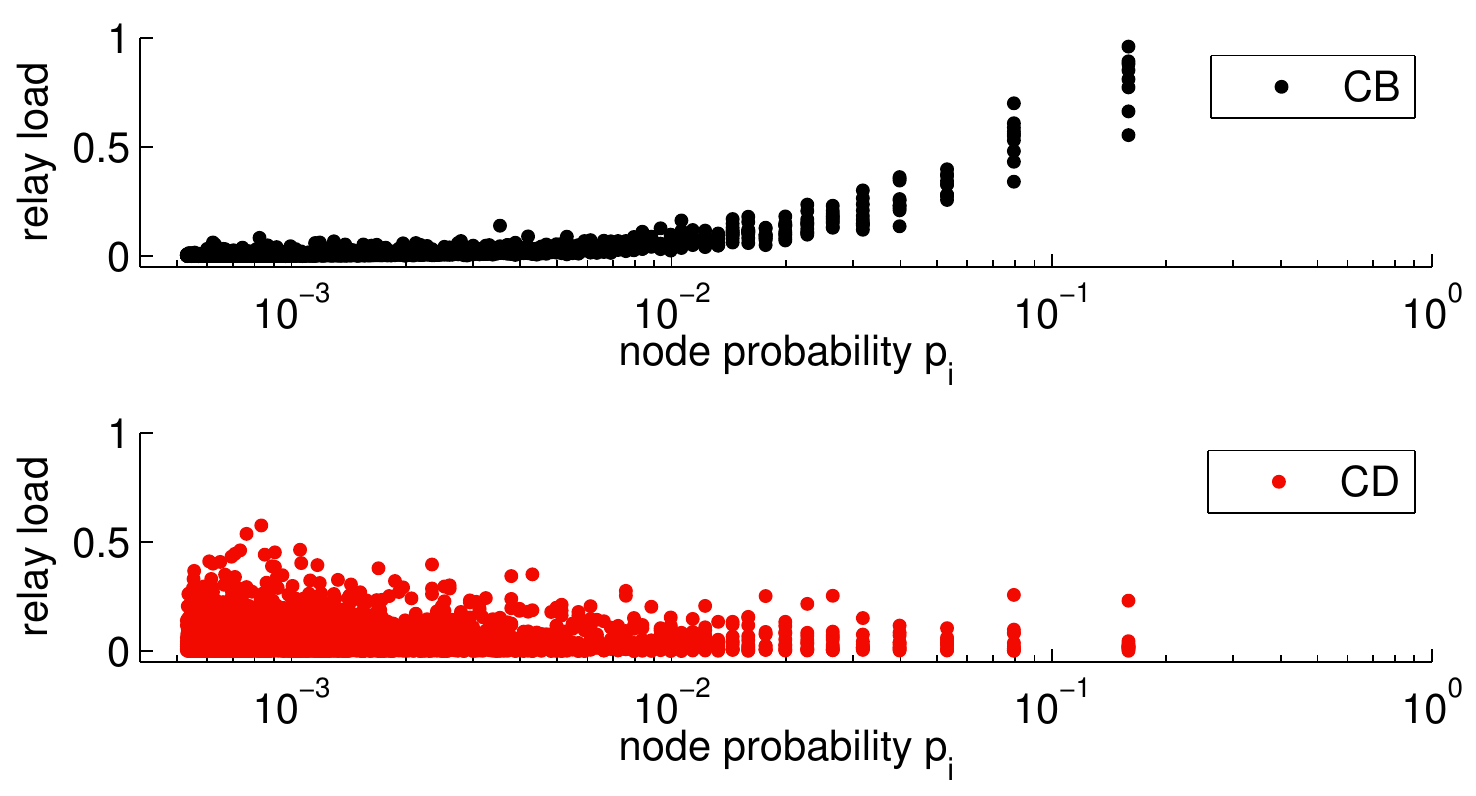}
	\caption{Node (weighted) relay load as function of node activity. In $CBR$ highly active node have high load a relays and low active nodes are not active as relays.
	In $CDR$ low active nodes my be have high load as relays.}
    \label{fig:nodeload}
\end{figure}

\paragraph{Robustness to Failures.}
Let us next study robustness properties.
We consider a scenario where 
every edge fails independently with probably $f$.
Then we run our routing algorithms with given message 
$TTL$s and a fail recovery mechanism. 
That is, in each step, 
the $TTL$ is decreased and the message 
is dropped if it did not reach 
the destination before $TTL=0$. 
To recover an edge failure (i.e., the next hop is not accessible), 
the message is sent to a random neighbor and re-routed from there
 to the destination. This mechanism is run in both the code and non-code based networks.

We first evaluate the \emph{route success fraction}, i.e., 
that fraction of routes that were to reach the destination for varying edge failure probabilities. 
Note that this is a \emph{weighted fraction}, and we sum the probabilities $p(i,j)$ of 
all the successful routes. In this simulation, 
we set the $TTL$ parameter to be $3$ times 
the route length without failures. 
The \emph{top subfigure} of Figure~\ref{fig:fail} shows our results for edges failure probabilities up to 50\%. Clearly, our approach improves the resistance to edge 
failures, and allows us to deliver more requests to the destination.
With $10\%$ failures, we can deliver more than $80\%$ of the 
messages and with $15\%$ we can deliver more than $70\%$ of the messages. 

The \emph{bottom subfigure} examine the effect of a maximum allowed routing length (TTL). 
We set the edge failure probability to be equal to $10\%$. 
The measure is the fraction of routing successes, which is equal to 
the \emph{accumulated} success of path weights $p(i,j)$ 
the figure shows that our approach is able to 
re-route traffic much better than the original approach,
as TTL increases. In the case where the TTL is 
greater than $3$ times the original path length, 
we succeed to finish routes in more than $80\%$, where in the
non code-based design the success rate in less than $50\%$.


\begin{figure}[t]
    \centering
	\includegraphics[width=0.9\linewidth]{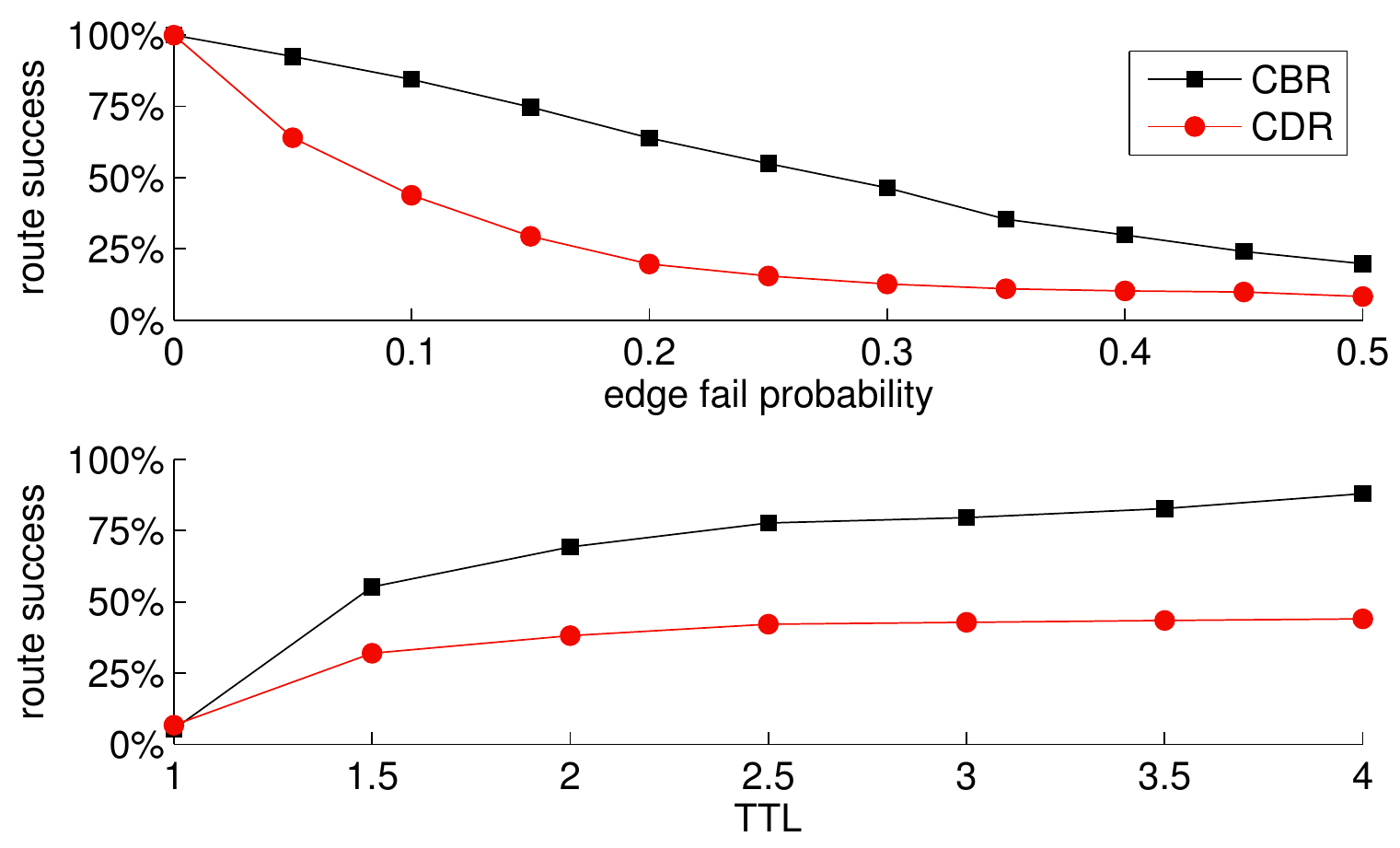}
	\caption{Even in the presence of severe edge failures, CBR consistently delivers more packets than CDR when TTL=3 time route length (\emph{top subfigure}). 
	The (\emph{bottom subfigure}) shows the results for increasing TTL and failure probability of $10\%$.}
	\label{fig:fail}
\end{figure}

\section{Conclusion}\label{sec:conclusion}

Inspired by the advent of more traffic optimized
datacenter interconnects and overlay topologies,
this paper introduced a formal metric and approach
to design robust and sparse network topologies 
providing information-theoretic path length guarantees,
based on coding. 
We see our contributions on the conceptual side and
understand our work as a first step. In particular, more work
is required to tailor our framework to a specific use case,
such as a datacenter interconnect or peer-to-peer overlay, where 
specific additional constraints may arise, such as channel
attenuation, jitter, etc., or where greedy routing may not be needed
but routes can be established along shortest paths.


\balance

{

}

\begin{thebibliography}{10}
\providecommand{\url}[1]{#1}
\csname url@samestyle\endcsname
\providecommand{\newblock}{\relax}
\providecommand{\bibinfo}[2]{#2}
\providecommand{\BIBentrySTDinterwordspacing}{\spaceskip=0pt\relax}
\providecommand{\BIBentryALTinterwordstretchfactor}{4}
\providecommand{\BIBentryALTinterwordspacing}{\spaceskip=\fontdimen2\font plus
\BIBentryALTinterwordstretchfactor\fontdimen3\font minus
  \fontdimen4\font\relax}
\providecommand{\BIBforeignlanguage}[2]{{%
\expandafter\ifx\csname l@#1\endcsname\relax
\typeout{** WARNING: IEEEtran.bst: No hyphenation pattern has been}%
\typeout{** loaded for the language `#1'. Using the pattern for}%
\typeout{** the default language instead.}%
\else
\language=\csname l@#1\endcsname
\fi
#2}}
\providecommand{\BIBdecl}{\relax}
\BIBdecl

\bibitem{comm-dc}
{M. Al-Fares et al.}, ``A scalable, commodity data center network
  architecture,'' in \emph{Proc. ACM SIGCOMM}, 2008.

\bibitem{projector}
{M. Ghobadi et al.}, ``{ProjecToR: Agile Reconfigurable Data Center
  Interconnect},'' in \emph{Proc. ACM SIGCOMM}, 2016.

\bibitem{splaynet}
{S. Schmid et al.}, ``Splaynet: Towards locally self-adjusting networks,''
  \emph{IEEE/ACM Trans. Netw. (ToN)}, 2016.

\bibitem{naor2007novel}
M.~Naor and U.~Wieder, ``Novel architectures for p2p applications: the
  continuous-discrete approach,'' \emph{ACM TALG}, 2007.

\bibitem{Leighton:1991:IPA:119339}
F.~T. Leighton, \emph{Introduction to Parallel Algorithms and Architectures:
  Array, Trees, Hypercubes}.\hskip 1em plus 0.5em minus 0.4em\relax Morgan
  Kaufmann Publishers Inc., 1992.

\bibitem{stoica2001chord}
I.~Stoica, R.~Morris, D.~Karger, M.~F. Kaashoek, and H.~Balakrishnan, ``Chord:
  A scalable peer-to-peer lookup service for internet applications,'' \emph{ACM
  SIGCOMM Computer Communication Review}, vol.~31, no.~4, pp. 149--160, 2001.

\bibitem{aspnes2007skip}
J.~Aspnes and G.~Shah, ``Skip graphs,'' \emph{ACM Transactions on Algorithms
  (TALG)}, vol.~3, no.~4, p.~37, 2007.

\bibitem{malkhi2002viceroy}
D.~Malkhi, M.~Naor, and D.~Ratajczak, ``Viceroy: A scalable and dynamic
  emulation of the butterfly,'' in \emph{Proceedings of the twenty-first annual
  symposium on Principles of distributed computing}.\hskip 1em plus 0.5em minus
  0.4em\relax ACM, 2002, pp. 183--192.

\bibitem{dcell}
C.~Guo, H.~Wu, K.~Tan, L.~Shi, Y.~Zhang, and S.~Lu, ``Dcell: A scalable and
  fault-tolerant network structure for data centers,'' in \emph{Proc. SIGCOMM},
  2008, pp. 75--86.

\bibitem{bcube}
C.~Guo, G.~Lu, D.~Li, H.~Wu, X.~Zhang, Y.~Shi, C.~Tian, Y.~Zhang, and S.~Lu,
  ``Bcube: A high performance, server-centric network architecture for modular
  data centers,'' in \emph{Proc. ACM SIGCOMM}, 2009, pp. 63--74.

\bibitem{rob-net-des-1}
A.~M. Koster, M.~Kutschka, and C.~Raack, ``Robust network design: Formulations,
  valid inequalities, and computations,'' \emph{Networks}, vol.~61, no.~2, pp.
  128--149, 2013.

\bibitem{rob-net-des-2}
N.~Olver and F.~B. Shepherd, ``Approximability of robust network design,''
  \emph{Mathematics of Operations Research}, vol.~39, no.~2, pp. 561--572,
  2014.

\bibitem{condor}
B.~Schlinker, R.~N. Mysore, S.~Smith, J.~Mogul, A.~Vahdat, M.~Yu,
  E.~Katz-Bassett, and M.~Rubin, ``Condor: Better topologies through
  declarative design,'' in \emph{Proc. ACM SIGCOMM}, 2015.

\bibitem{traffic-1}
T.~Benson, A.~Akella, and D.~A. Maltz, ``Network traffic characteristics of
  data centers in the wild,'' in \emph{Proc. 10th ACM SIGCOMM Conference on
  Internet Measurement (IMC)}, 2010, pp. 267--280.

\bibitem{ref16-helios}
N.~Farrington, G.~Porter, S.~Radhakrishnan, H.~H. Bazzaz, V.~Subramanya,
  Y.~Fainman, G.~Papen, and A.~Vahdat, ``Helios: A hybrid electrical/optical
  switch architecture for modular data centers,'' in \emph{Proc. ACM SIGCOMM},
  2010, pp. 339--350.

\bibitem{ref25-reactor}
H.~Liu, F.~Lu, A.~Forencich, R.~Kapoor, M.~Tewari, G.~M. Voelker, G.~Papen,
  A.~C. Snoeren, and G.~Porter, ``Circuit switching under the radar with
  reactor,'' in \emph{Proc. 11th USENIX Conference on Networked Systems Design
  and Implementation (NSDI)}, 2014, pp. 1--15.

\bibitem{ref23-flyways}
S.~Kandula, J.~Padhye, and P.~Bahl, ``Flyways to de-congest data center
  networks,'' 2009.

\bibitem{zhou2012mirror}
X.~Zhou, Z.~Zhang, Y.~Zhu, Y.~Li, S.~Kumar, A.~Vahdat, B.~Y. Zhao, and
  H.~Zheng, ``Mirror mirror on the ceiling: flexible wireless links for data
  centers,'' \emph{ACM SIGCOMM Computer Communication Review}, vol.~42, no.~4,
  pp. 443--454, 2012.

\bibitem{ref22-firefly}
N.~Hamedazimi, Z.~Qazi, H.~Gupta, V.~Sekar, S.~R. Das, J.~P. Longtin, H.~Shah,
  and A.~Tanwer, ``Firefly: A reconfigurable wireless data center fabric using
  free-space optics,'' in \emph{Proc. SIGCOMM}, 2014, pp. 319--330.

\bibitem{disc17netdesign}
C.~Avin, K.~Mondal, and S.~Schmid, ``Demand-aware network designs of bounded
  degree,'' in \emph{ArXiv Technical Report}, 2017.

\bibitem{erdos1959on-random}
P.~Erd\H{o}s and A.~R\'enyi, ``On random graphs,'' \emph{Publicationes
  Mathemticae (Debrecen)}, vol.~6, pp. 290--297, 1959.

\bibitem{kleinberg2000small}
J.~Kleinberg, ``The small-world phenomenon: An algorithmic perspective,'' in
  \emph{Proceedings of the thirty-second annual ACM symposium on Theory of
  computing}.\hskip 1em plus 0.5em minus 0.4em\relax ACM, 2000, pp. 163--170.

\bibitem{shannon2001mathematical}
C.~E. Shannon, ``A mathematical theory of communication,'' \emph{ACM SIGMOBILE
  Mobile Computing and Communications Review}, vol.~5, no.~1, pp. 3--55, 2001.

\bibitem{thomas2006elements}
T.~M. Cover and J.~A. Thomas, \emph{Elements of information theory}.\hskip 1em
  plus 0.5em minus 0.4em\relax Wiley New York, 2006, ch.~5, pp. 127--128.

\bibitem{bruijn1946combinatorial}
D.~N. Bruijn, ``A combinatorial problem,'' \emph{Proc. Koninklijke Nederlandse
  Akademie van Wetenschappen. Series A}, vol.~49, no.~7, p. 758, 1946.

\bibitem{hoory2006expander}
S.~Hoory, N.~Linial, and A.~Wigderson, ``Expander graphs and their
  applications,'' \emph{Bulletin AMS}, 2006.

\end{thebibliography}
\end{document}